\documentclass[aos,preprint]{imsart}

\usepackage{}

\RequirePackage[OT1]{fontenc}
\RequirePackage{amsthm,amsmath}
\RequirePackage[numbers]{natbib}
\RequirePackage[colorlinks,citecolor=blue,urlcolor=blue,pdfstartview=FitB]{hyperref}

\arxiv{math.PR/0000000}

\usepackage{latexsym,amsmath,amsfonts}
\usepackage{algorithmicx,algpseudocode}
\usepackage{amssymb}
\usepackage{pstricks,pst-node,pst-tree}
\usepackage{thmtools,thm-restate}
\usepackage{graphicx,epsfig,epstopdf}
\usepackage{multirow}

\startlocaldefs
\numberwithin{equation}{section}
\theoremstyle{plain}

\newtheorem{theorem}{Theorem}

\newtheorem{corollary}{Corollary}
\newtheorem{lemma}{Lemma}

\def\Ep{{\rm E}}
\def\En{{\mathbb{E}_n}}
\def\Gn{{\mathbb{G}_n}}
\def\P{{\rm P}}

\renewcommand{\Pr}[1]{\mathbb{P}\left(#1\right)} 

\def\supp{{\rm supp}}

\endlocaldefs

\begin{document}

\begin{frontmatter}
\title{High-dimensional Linear Models with Many Endogenous Variables}
\runtitle{Inference with Many Endogenous Variables}
\thankstext{T1}{Current Version: March 28, 2019.}

\begin{aug}
\author{\fnms{Alexandre} \snm{Belloni}\ead[label=e1]{abn5@duke.edu}},
\author{\fnms{Christian} \snm{Hansen}\ead[label=e3]{Christian.Hansen@chicagobooth.edu}}
\and
\author{\fnms{Whitney} \snm{Newey}\ead[label=e4]{wnewey@mit.edu}}




\end{aug}

\begin{abstract}
High-dimensional linear models with endogenous variables play an increasingly important role in recent econometric literature. In this work we allow for models with many endogenous variables and many instrument variables to achieve identification. Because of the high-dimensionality in the second stage, constructing honest confidence regions with asymptotically correct coverage is non-trivial. Our main contribution is to propose estimators and confidence regions that would achieve that.

The approach relies on moment conditions that have an additional orthogonal property with respect to nuisance parameters. Moreover, estimation of high-dimension nuisance parameters is carried out via new pivotal procedures. In order to achieve simultaneously valid confidence regions we use a multiplier bootstrap procedure to compute critical values and establish its validity.
\end{abstract}



\end{frontmatter}

\section{Introduction}

We consider estimation of parameters in high-di\-men\-sional linear regression models with many endogenous variables:
\begin{align*}
y_i = x_i'\beta_0 + \xi_i, \quad \Ep[x_i \xi_i] \ne 0.
\end{align*}
We allow for cases where $p$, the dimension of the endogenous variable $x_i$, is large relative to the sample size $n$; and we consider cases where many or all of the parameters in $\beta$ are of interest.  This structure arises naturally in many structural economic models.  For example, one may be interested in estimating a demand relationship in which the quantity demanded of a given product may depend on its own price as well as the price of other related products.  In this case, $y_i$ represents the ($\log$) quantity sold of the product of interest, $x_i$ represents the vector of ($\log$) prices of all potentially related products including the price of the product of interest, and $i$ indexes independent markets in which data on these products are available.  $\beta_0$ then provides information on own- and cross-price elasticities for a given product.

Of course, informative estimation and inference about $\beta$ is impossible without further structure.  We adopt an instrumental variables (IV) framework and assume that a set $z_i$ of $K$ instruments with $K \geq p$ is available such that $\Ep[z_i\xi_i] = 0$.  We thus work in a model with many endogenous variables and many instruments.

To cope with the high-dimensionality of $\beta_0$, we adopt a sparse modeling framework and assume that the number of non-zero elements in $\beta_0$ is small relative to the sample size.  We also make use of sparsity assumptions over first-stage relationships. Specifically, suppose we are interested in a particular element, $\beta_{0,j}$, of the vector $\beta_0$.  We define an instrument specific to $\beta_{0,j}$ as a linear combination $z_i'\mu_j$ of the available high-dimensional set of instruments that satisfies the usual IV moment restriction $\Ep[(y_i - \beta_{0,j}x_{j,i} - \sum_{k \ne j}\beta_{0,k}x_{k,i})(z_i'\mu_j)] = 0$ and require that $\mu_j$ is a sparse vector.  When we are interested in multiple elements of $\beta_0$, we form a different linear combination of the available set of instruments for each element of interest and require that all such linear combinations are sparse.  We also note that the rank condition for identification, which we maintain, will implicitly restrict the association between each of these constructed linear combinations.

In principle, there are many ways to choose the $\mu_j$ used to define the linear combinations of instruments we use to identify the $\beta_{0,j}$.  We exploit this to choose the $\mu_j$ to guarantee a type of orthogonality between the constructed instrument $z_i'\mu_j$ and the nuisance function in the moment condition for $\beta_{0,j}$, $\sum_{k \ne j}\beta_{0,k}x_{k,i}$.  Specifically, we choose the $\mu_j$ so that $\Ep[x_{i,-j}z_i'\mu^j_0] = 0$ while requiring that $\Ep[x_{i,j}z_i'\mu^j_0] \ne 0$.  The second condition is simply the usual instrument relevance condition.  The first condition ensures that the moment condition used to identify $\beta_{0,j}$ is locally insensitive, or orthogonal, to the value of the high-dimensional set of nuisance parameters $\{\beta_{0,k}\}_{k \ne j}$.  This orthogonality is a key ingredient to producing post-model-selection estimators that are $\sqrt{n}$-consistent and asymptotically normal and post-model-selection inferential procedures that are uniformly valid across a wide range of data generating processes (DGPs) that include cases where perfect model selection is theoretically impossible.  For further detail and discussion, see, for example, \cite{BCFH:Policy}, \cite{CHS:AnnRev}, and \cite{DML}.

We provide a simple-to-implement multi-step estimation and inference procedure for the structural parameters.  We obtain initial estimates of $\beta_0$ by adapting the estimator of \cite{gautier2011high}.  We then obtain estimates of the $\mu_{0,j}$ through a novel estimation procedure that builds upon \cite{BCKRT2016a}.  Given these estimates, we construct the empirical analog of the IV moment condition for $\beta_{0,j}$ by plugging in the estimated values for $\{\beta_{0,k}\}_{k \ne j}$ and $\mu_{0,j}$ and solve for the final estimator of $\beta_{0,j}$.  The final step of the procedure is trivial, and the first two steps proceed by solving high-dimensional convex optimization problems for which efficient optimization routines are available.  A major novelty of our approach is that it is formulated in such a way that we may use simple, nuisance parameter free values for the penalty parameters involved in the two penalized optimization problems even when model errors are heteroskedastic.  Thus, we avoid potentially complicated and theoretically problematic data-dependent choice of tuning parameters.  The proposed high-dimensional estimation procedure may be of substantive interest outside of the present context.  Finally, we employ the multiplier bootstrap for inference about parameters of interest in $\beta_0$.  Following, e.g., \cite{chernozhukov2013gaussian}, we establish that the multiplier bootstrap provides valid simultaneous confidence regions allowing for a high-dimensional set of parameters of interest; see also \cite{Buehlmann:etal:bootstrap:2016} for an application in a linear model with exogenous right-hand-side variables.

Our paper contributes to the growing literature on inference for structural parameters in sparse high-dimensional IV settings.  Much work in this setting focuses on the case where the dimension of the endogenous variable is small but there are a large number of available instruments; e.g. \cite{BaiNg2009a}, \cite{BaiNg2010}, \cite{BellChernHans:Gauss}, \cite{BellChenChernHans:nonGauss}, \cite{bch:jep}; or a large number of instruments and controls; e.g. \cite{CHS:PnP}, \cite{CHS:AnnRev}.  \cite{BCFH:Policy} considers estimation of treatment effects in IV models with binary instrument and endogenous variable in the presence of a high-dimensional set of control variables in detail.  In a different but tangentially related direction, \cite{kangetal1} and \cite{Caner:invalid}, among others, consider estimation in the presence of a high-dimensional set of instruments allowing for some of the instruments to be invalid with data-dependent selection of the valid and invalid instruments.  The work in the present paper is most closely related to \cite{fan:liao} and especially \cite{gautier2011high}, both of which consider estimation in models with a large number of endogenous variables.  Both papers adopt similar estimation strategies to that employed in this paper though neither explicitly constructs instruments to guarantee orthogonality and \cite{fan:liao} focuses on establishing oracle estimation results rather than providing a uniformly valid inferential procedure.

Regarding uniform valid inference, the work of \cite{neykov2015unified} proposes a unifying treatment for various models with linear moment conditions based in de-biasing ideas and many endogenous variables are discussed in the appendix as one of the many examples. Another related work is \cite{gold2017inference}, developed concurrently with this work, proposes the use of related de-biasing ideas after the use of lasso on estimated regressors (the projection of the endogenous variables on the instruments) to obtain initial estimates of the coefficients of these endogenous variables. This has the nice feature to allow the authors to invoke the regularity conditions on the design matrix that are standard in the literature (e.g. restricted eigenvalue conditions) but additional penalization is needed to handle the error in the estimated regressors. We also note that the work \cite{gautier2011high} is currently being revised and also pursues uniform valid inference building upon ideas from \cite{javanmard2014confidence}.

Our paper provides a useful complement to the existing high-dimensional IV work by allowing for a high-dimensional set of endogenous variables, providing an easy-to-implement estimation procedure with a simple pivotal tuning parameter choice, and establishing a uniformly valid inferential procedure for a high-dimensional set of coefficients of interest within the linear IV model. We explicitly allow non-i.i.d. observations, non-exponential tails, provide procedures where penalty parameters are pivotal, and confidence regions that are simultaneously valid for many parameters. An important technical contribution of our work is a new bound on the sensitivity quantity that was proposed in \cite{gautier2011high}. We note that in Section 9 of \cite{gautier2011high}, a related pivotal problem was proposed to identify invalid moment conditions when an initial estimator $\hat\beta$ is available. The pivotal estimation for the orthogonalization/de-biasing step seems to be new as well.

The remainder of the paper is organized as follows.  We first restate the model and provide a detailed outline of the estimation approach in Section \ref{sec:Model}.  In Section \ref{sec:Main}, we provide conditions under which we verify that the proposed estimator has desirable properties and under which we can produce inferential statements that are uniformly valid over a large class of models.  We present proofs of the main results in the appendix.

\section{Model and Method}\label{sec:Model}

We consider a linear instrumental variable model
\begin{equation}\label{linmodel:iv} y_i = x_i'\beta_0 + \xi_i, \ \ \ \ \Ep[x_i\xi_i]\neq 0 \end{equation}
where potentially many of the components of $x$ are correlated with the noise term $\xi$ where a $K$-dimensional vector of valid instruments $z$ is available. Valid instruments satisfy the following orthogonal condition
\begin{equation}\label{linmodel:iv2}  \Ep[ z(y-x'\beta_0)]=0 \end{equation}
which can be used to identify $\beta_0$ under appropriated conditions. (Therefore if the $j$th component of $x$ is not correlated with the noise, $x_j$ is a valid instrument.) It follows that we need at least as many instruments as the dimension of $\beta_0$ (i.e. $K\geq p$).

To construct estimators and associated honest confidence regions, we will construct an instrument tailored for each component $\beta_{0j}$ of interest. More formally we will consider a linear combination of the original instruments $z'\mu$, to create a moment condition whose unique solution is $\theta=\beta_{0j}$
$$ \Ep[(y-x_j\theta-x_{-j}'\beta_{0,-j})z'\mu^j]=0 $$
where $\mu^j$ is a non-zero $K$-dimensional vector. It follows that there are many choices of $\mu^j$ that satisfy the relation above by (\ref{linmodel:iv2}). We exploit the freedom to pick $\mu^j$ to choose $\mu^j_0$ so that an orthogonality condition also holds, namely
\begin{equation}\label{ortho:new}  \Ep\left[\frac{1}{n}\sum_{i=1}^n x_{i,-j}z_i' \right]\mu^j_0 = 0 \end{equation}
while having $z'\mu^j_0$ correlated with $x_j$, i.e. $\Ep\left[\frac{1}{n}\sum_{i=1}^n x_{ij}z_i'\mu^j_0 \right] \neq 0$. Such orthogonality condition plays a key role in our results to reduce the impact of estimation mistakes of $\beta_{0,-j}$. We will define $\mu^0_j$ as
$$ (\mu^j_0,\vartheta^j_0)= \arg\min_{\mu,\vartheta} \Ep[\mbox{$\frac{1}{n}\sum_{i=1}^n$} (x_{ij}-z_i'\mu-x_{i,-j}'\vartheta)^2]  \ : \ \Ep[\mbox{$\frac{1}{n}\sum_{i=1}^n$}x_{-j}z_i'\mu]=0. $$
In order to construct our estimator for $\beta_{0j}$, we will need estimates of the nuisance parameters $\eta^j_0=(\beta_{0,-j},\mu^j_0)$ which are attainable under sparsity conditions to cope with their high-dimensionality.

Estimators for $\beta_0$ that can achieve good rates of convergence have been proposed in the literature \cite{gautier2011high}. Here we propose to use
\begin{equation}\label{def:nuisancebeta-j}
\begin{array}{rl}
(\hat\beta, \hat t) \in \arg\min_{\beta,t} & \|\beta\|_1 + \lambda_t\|t\|_\infty\\
& | \ \En[ (y-x'\beta)z_\ell ]\ | \leq \tau t_\ell, \ \ \ell \in [K]\\
& \{\En[ (y-x'\beta)^2z_\ell^2]\}^{1/2} \leq t_\ell, \ \ \ell \in [K]\\
\end{array}
\end{equation}

In order to estimate a $\mu^j_0$ we will also build upon these self-normalization ideas to achieve a procedure that does not require to tune parameters. However additional constrains are necessary to achieve the orthogonality condition (\ref{ortho:new}). Our proposed estimator seems to be new and can be of independent interest in other setting:
\begin{equation}\label{def:nuisancej}
\begin{array}{rl}
(\hat\mu^j,\hat\vartheta^j,\hat t^j) \in & \arg  \min_{\mu,t,\vartheta}  \|\mu\|_1 + \|\vartheta\|_1 + \lambda_t\|t\|_\infty\\
&  \ v_{ij}=x_{j}-z'\mu-x_{-j}'\vartheta\\
& \ | \ \En[ v_{j}z_\ell ]\ | \leq c\tau t_\ell^z, \ \  \{\En[ v_{j}^2z_\ell^2]\}^{1/2} \leq t_\ell^z, \ \ \ell\in [K]\\
& \ | \ \En[ v_{j}x_\ell ]\ | \leq c\tau t_\ell^x, \ \  \{\En[ v_{j}^2x_\ell^2]\}^{1/2} \leq t_\ell^x, \ \ \ell\in [p]\setminus\{j\}\\
& \ | \ \En[ x_{\ell}z'\mu ]\ | \leq c\tau t_\ell^{xz}, \  \{\En[ \{x_{\ell}z'\mu\}^2]\}^{1/2} \leq t_\ell^{xz}, \ \ \ell \in [p]\setminus\{j\}\\
\end{array}
\end{equation} where $c\geq 1$ is a fixed constant.\footnote{This constant accounts for non-i.i.d. setting. We can take $c=1$ if the data is i.i.d..}

Based on the estimates $\hat\eta^j=(\hat\beta_{-j},\hat\mu^j)$ of $\eta^j_0=(\beta_{0,-j},\mu^j_0)$ based on (\ref{def:nuisancebeta-j}) and (\ref{def:nuisancej}), we compute our estimate $\check\beta_j$ as the solutions of
\begin{equation}\label{ortho:j}
\begin{array}{rl}
0 = \frac{1}{n}\sum_{i=1}^n(y_i-x_{ij}\theta-x_{i,-j}\hat\beta_{-j})z_i'\hat\mu^j. \end{array}\end{equation}
It follows the estimator can be computed in closed form as
\begin{equation}\label{def:checkbeta}\check\beta_j =  \frac{\hat\Omega_j^{-1}}{n}\sum_{i=1}^n(y_i-x_i'\hat \beta_{-j})z_i'\hat\mu^j\end{equation}
where $\hat\Omega_j = \frac{1}{n}\sum_{i=1}^nx_{ij}z_i'\hat\mu^j$ is an estimate of $\Omega_j=\Ep[\frac{1}{n}\sum_{i=1}^nx_{ij}z_i'\mu^j_0]$.

Next we construct simultaneous confidence regions for the coefficients corresponding to all components in a set $ S\subseteq\{1,\dots,p\}$, i.e. for a given $\alpha\in(0,1)$, we choose a critical value $c_{\alpha,S}^*$ such that with probability converging to $1-\alpha$ we have
\begin{equation}\label{region:simultaneous} \check \beta_j - c_{\alpha,S}^* \frac{\hat\sigma_j}{\sqrt{n}}\leq \beta_{0j} \leq \check \beta_j + c_{\alpha,S}^* \frac{\hat\sigma_j}{\sqrt{n}} \ \ \mbox{for all} \ j\in S \end{equation}
where the  estimate of $\sigma_j^2 = \Omega_j^{-2}\Ep[\frac{1}{n}\sum_{i=1}^n\{(y_i-x_i'\beta_0)z_i'\mu^j_0\}^2]$ is given by
\begin{equation}\label{est:sigmaj}
\hat\sigma_j^2 = \hat\Omega_j^{-2}\frac{1}{n}\sum_{i=1}^n\{(y_i-x_i'\hat\beta)z_i'\hat\mu^j\}^2.\end{equation}

We apply the procedure (\ref{def:nuisancej})-(\ref{def:checkbeta}) to each component $j\in S$ to obtain $\check\beta_j$ and use a multiplier bootstrap to compute the appropriate critical value in (\ref{region:simultaneous}). The confidence bands have the desired asymptotic coverage even in cases where the cardinality of $S$ is larger than the sample size $n$ as our analysis builds upon results for high-dimensional central limit theorems developed in  \cite{chernozhukov2013gaussian,chernozhukov2014clt,chernozhukov2012gaussian,chernozhukov2012comparison,chernozhukov2015noncenteredprocesses}. We define the vector $\widehat{\mathcal{G}}$ as
\begin{equation}\label{def:mb} \widehat{\mathcal{G}}_j:=-\frac{1}{\sqrt{n}}\sum_{i=1}^n g_i \hat \sigma_j^{-1}\hat\Omega_j^{-1} (y_i-x_i'\hat\beta)z_i'\hat\mu^j, \ \ j \in S \end{equation}
where $(g_i)_{i=1}^n$ are independent standard normal random variables independent from $(y_i,x_i,z_i)_{i=1}^n$. We compute the critical value $c_{\alpha,S}^*$ as the $(1-\alpha)$-quantile of the conditional distribution of $$\max_{j\in S} |\widehat{\mathcal{G}}_j|$$ given the data. For the readers convenience Appendix A provides a detailed description of the algorithm to be implemented for the proposed estimators and confidence regions.

\section{Theoretical Results}\label{sec:Main}

Next we provide regularity conditions on the data generating process (\ref{linmodel:iv})-(\ref{linmodel:iv2}) under which the estimators and confidence regions proposed in Section \ref{sec:Model} are guaranteed to perform well despite model selection mistakes and high-dimensionality of nuisance parameters.

\subsection{Rates of Convergence of $\hat\beta$ in (\ref{def:nuisancebeta-j})}\label{Sec:Step1}

We first establish bounds for the $\ell_q$-error of the estimator $\hat\beta$ defined as in (\ref{def:nuisancebeta-j}). Such bounds will yield rates of convergence that are useful to the development of confidence regions (\ref{region:simultaneous}) since the procedure relies on having estimators of the nuisance parameters. In order to analyze this estimator we consider the $\ell_q$-sensitivity coefficients proposed in \cite{gautier:tsybakov} which account for the instruments. For $q\geq 1$, and  $\Psi = \frac{1}{n}Z'X$, define
\begin{equation}\label{defkappa}
 \kappa_q^{ZX}(s,u) = \min_{J:|J|\leq s} \left( \min_{\theta \in C_J(u):\|\theta\|_q=1} \|\Psi \theta \|_\infty \right) \end{equation}
where $C_J(u)=\{\theta \in \mathbb{R}^p : \|\theta_{J^c}\|_1 \leq u \|\theta_J\|_1\}$, $u>0$ and $J\subseteq \{1,\ldots,p\}$. These sensitivity quantities are useful to establish convergence in the $\ell_q$-norms for penalized estimators as shown in \cite{gautier2013pivotal,gautier2011high}. Next we establish new lower bounds on the $\ell_q$-sensitivity coefficients for $q\in\{1,2\}$ that are of independent interest. These lower bounds lead to sharp rates for many design of interest in econometric settings.

\begin{theorem}\label{thm:ExtraSensitivity}  We have that for $q\in \{1,2\}$
$$\begin{array}{rl}
\kappa_q^{\Psi}(s,u)
& \geq \max_{m \geq s } \left\{\frac{\sigma_{\min}(m)}{\sqrt{m}} - \frac{\sigma_{\max}(m)}{\sqrt{m}}(1+u)\sqrt{s/m}\right\}\frac{s^{1/2-1/q}}{\{1+(1+u)\sqrt{s/m}\}(1+u)} \\
\end{array}
$$ where $$\sigma_{\min}(m):=\min_{|M|\leq m}\max_{|J|\leq m} \sigma_{\min}(\Psi_{J,M}),$$ $$\sigma_{\max}(m):=\max_{|M|\leq m}\max_{|J|\leq m} \sigma_{\max}(\Psi_{J,M}),$$
 where $\sigma_{min}(A)$ is the smallest singular value of the matrix $A$, and $\sigma_{\max} (A)$ denotes the maximum singular eigenvalue of $A$, and $A_{J,M}$ denotes a submatrix $\{A_{j,m}\}$
where $j \in J$ and $m \in M$.
\end{theorem}

Theorem \ref{thm:ExtraSensitivity} extends the scope of designs for which we have lower bounds. The main difficulty is to handle non-symmetric matrices $\Psi$ with off-diagonal elements that are not small. The proof relies on sparse singular values decomposition (via $\sigma_{\min}(m)$ and $\sigma_{\max}(m)$) which are akin to sparse eigenvalues that were used in the literature, see e.g. \cite{BickelRitovTsybakov2009,BellChenChernHans:nonGauss}. The following corollaries provides bounds under primitive conditions. The result of Theorem \ref{thm:ExtraSensitivity} can also be applied in the presence of weak instruments as illustrated below.

\begin{corollary}\label{cor:ExtraSensitivity-WeakIV}
Suppose that the $K \times p$ matrix  $\Psi$ has  $m$-sparse singular values bounded from below by $\mu_n>0$ and from above by $1$ for $m = (64s/\mu_n^2)$. Then, the matrix $\Psi$ satisfies
$$ \kappa_q^{\Psi}(s,3)  \geq  c\mu_n^2 s^{-1/q} \ \ q\in\{1,2\}$$
for some universal constant $c>0$.
\end{corollary}

The next corollary considers the case of random design with strong instruments.

\begin{corollary}\label{cor:ExtraSensitivity}
Suppose that $(z_i,x_i)_{i=1}^n$ are i.i.d. random vectors such that $\Ep[\max_{i\in[n]}\|(z_i,x_i)\|_\infty^{2}] \leq K_n$, and that $\Ep[z_ix_i']$ has $s(1+u)^2\log n$-sparse singular values bounded from below by $c>0$ and from above by $C<\infty$. Then, with probability $1-o(1)$, as $n$ grows, the matrix $\Psi=\frac{1}{n}\sum_{i=1}^nz_ix_i'$ satisfies
$$ \kappa_q^{\Psi}(s,u)  \geq  \frac{\tilde cs^{-1/q}}{(1+u)^2} \ \ q\in\{1,2\}$$
provided that $(1+u)^6s^2 K_n \log(Kp) = o(n)$ where the constant $\tilde c$ depends only on $c$ and $C$.
\end{corollary}

Under conditions of Corollary \ref{cor:ExtraSensitivity} and the typical case that $u \leq 3$, we have \begin{equation}\label{conditionkappa} \kappa_q^{\Psi}(s,3)  \geq s^{-1/q} c'\end{equation} with high probability. In turn, under Condition IV1 below, (\ref{conditionkappa}) allows the pivotal estimator proposed in (\ref{def:nuisancebeta-j}) to achieve optimal rates. In the analysis of subsequent sections we will assume  directly  (\ref{conditionkappa}) given that Corollary \ref{cor:ExtraSensitivity} provides an interesting set of data generating processes that imply it with high probability.

~\\
{\bf Condition IV1.} {\it (i) The vectors $\{(y_i,z_i,x_i,\xi_i), i=1,\ldots,n\}$ are independent across i obeying (\ref{linmodel:iv}) and (\ref{linmodel:iv2}) with $\|\beta_0\|_0 \leq s$. (ii) Suppose that the sequence  $M_{nz\xi}:=\max_{j\leq K}\Ep[\frac{1}{n}\sum_{i=1}^n|z_{ij}\xi_i|^3]^{1/3} / \Ep[\frac{1}{n}\sum_{i=1}^n|z_{ij}\xi_i|^2]^{1/2}$ satisfies $M_{nz\xi}\log^{1/2}(npK) = o(n^{1/6})$ and define $H_n:=\max_{j\leq K} \|\frac{1}{n}\sum_{i=1}^nz_{ij}^2 x_ix_i'\|_{\infty}$.}
~\\

Condition IV1 is a standard condition  allows for heteroskedasticity. It restricts the  number of covariates and instruments, but it allows for them to be larger than the sample size. The next theorem establishes bounds on the estimation error $\hat\beta-\beta_0$.

\begin{theorem}\label{thm:linmodel:iv}
Under Condition IV1, setting the parameters $\lambda_t = 1/(2H_n)$ and $\tau = n^{-1/2}\Phi^{-1}(1-\alpha/(2p))$, for $\alpha \in (n^{-1},1)$, provided $\tau \leq \frac{1}{2}\lambda_t\kappa_1^{ZX}(s,3)$ holds with probability at least $1-\epsilon$, we have that
with probability $1-\alpha\{1+o(1)\}-\epsilon$ the estimator defined in (\ref{def:nuisancebeta-j}) satisfies
$$ \|\hat\beta-\beta_0\|_q \leq 4\tau \frac{\max_{j\leq K}\{\frac{1}{n}\sum_{i=1}^nz_{ij}^2\xi_i^2\}^{1/2}}{\kappa_q^{ZX}(s,3)}.$$
\end{theorem}

Theorem \ref{thm:linmodel:iv} has several important features. First it allows for heteroskedasticity which is prevalent in econometric applications. Second, the definition of the estimator (\ref{def:nuisancebeta-j}) uses penalty parameters that do not depend on any unknown quantity. In particular, it bypasses the need for cross validation which can be effective in some cases but validity is not typically established in the literature. This approach builds upon self-normalized ideas that were applied in \cite{BCKRT2016a} to error-in-variables models.

\subsection{Rates of Convergence of $\hat\mu^j$ in (\ref{def:nuisancej})}\label{Sec:Step2}

Next we consider the estimation of the nuisance parameter $\mu^j_0$ for all $j\in S$.  Our analysis still also relies on $\ell_q$-sensitivity coefficients proposed in \cite{gautier:tsybakov} but for different choice of design matrix given by $\Psi_j = \frac{1}{n}\widetilde Z' \widetilde Z$ where $\widetilde Z=(Z, X_{-j})$ and further restricting to directions $\delta \in C_J(u)$. We define
\begin{equation}\label{defkappa}
 \kappa_q^{\widetilde Z\widetilde Z}(s,u) = \min_{j\in S} \min_{J:|J|\leq s} \left( \min_{\theta \in \widetilde C_{j,J}(u):\|\theta\|_q=1} \|\Psi_j \theta \|_\infty \right) \end{equation}
where $\widetilde C_{j,J}(u) = \{ \theta \in C_J(u) : \exists \theta^1,\theta^2,  \ \ \theta=\theta^1-\theta^2, \ \  \|\En[x_{-j}z'\theta^m]\|_\infty\leq 2\tau\max_{k\in[p]\setminus\{j\}}\{\En[x_k^2(z'\theta^m)^2]\}^{1/2}, m=1,2\}$. In many designs we can establish lower bounds on $\kappa_q^{\widetilde Z\widetilde Z}(s,u)$ ignoring the additional restrictions.

In what follows we define $v_{ij}:=x_{ij}-z_i'\mu^j_0-x_i'\vartheta_0^j$, $j\in S$. Define:\\
 $\bar M_{nvz} := \Ep[\max_{i\leq n,j\in S,k}|v_{ij}z_{ik}\Ep[v_{ij}z_{ik}]|^2]^{1/2}$,\\
 $\bar M_{nvx} :=\Ep[\max_{i\leq n,j\in S,k}|v_{ij}x_{ik}\Ep[v_{ij}x_{ik}]|^2]^{1/2}$,\\
 $\bar M_{nxz} :=\Ep[\max_{i\leq n,j\in S,k}|x_{ik}z_i'\mu_0^j\Ep[x_{ik}z_i'\mu_0^j]|^2]^{1/2}$, and\\
  $\bar M_n := \bar M_{nxz}\vee \bar M_{nvz}\vee \bar M_{nvx}$. \\
  Note that in the i.i.d. setting we have $\bar M_n=0$.

~\\
{\bf Condition IV2.} {\it For each $j \in S$: (i) The vectors $\{(z_i,x_i), i=1,\ldots,n\}$ are independent across $i\in [n]$ obeying (\ref{linmodel:iv}), (\ref{linmodel:iv2}), and (\ref{ortho:new}) with $\|\mu^j_0\|_0 \leq s$. (ii) We assume: ${\rm Var}(\frac{1}{\sqrt{n}}\sum_{i=1}^nx_{ik}z_i'\mu_0^j)^{1/2}\geq \tilde c$, ${\rm Var}(\frac{1}{\sqrt{n}}\sum_{i=1}^nv_{ij}z_{ik})^{1/2}\geq \tilde c$, ${\rm Var}(\frac{1}{\sqrt{n}}\sum_{i=1}^nv_{ij}x_{i\ell})^{1/2}\geq \tilde c$, $\Ep[\frac{1}{n}\sum_{i=1}^n|x_{ik}z_i'\mu_0^j|^3] \leq C$, $\Ep[\frac{1}{n}\sum_{i=1}^n|v^j_iz_{ik}|^3]^{1/3}\leq C$, and $\Ep[\frac{1}{n}\sum_{i=1}^n|v_{ij}x_{ik}|^3]^{1/3}\leq C$. Furthremore,
 It holds that $\log^{1/2}(pnK)=o(n^{1/6})$ and $Cn^{-1}\bar M_n\log(Kpn)=o(1)$.
Also define $H_n:= H_n^x\vee H_n^z\vee H_n^{xz}$ where $H_n^x := \max_{k\in[K]} \|\frac{1}{n}\sum_{i=1}^nz_{ik}^2 x_ix_i'\|_{\infty}$, $H_n^x := \max_{\ell \in [p]} \|\frac{1}{n}\sum_{i=1}^nx_{i\ell}^2 z_iz_i'\|_{\infty}$, and $H_n^{z} := \|\frac{1}{n}\sum_{i=1}^nz_{ik}^2 z_iz_i'\|_{\infty}.$}
~\\

Condition IV2 contains a mild condition on the design,   is a standard condition  allows for heteroskedasticity and fixed design. The next theorem summarizes the properties of the estimator defined in (\ref{def:nuisancej}).

\begin{theorem}\label{thm:nuisancej}
Under Condition IV2, setting the parameters $\lambda_t = 1/(2H_n)$ and $\tau = 1.1n^{-1/2}\Phi^{-1}(1-\alpha/(2|S|(K+p))$, for $\alpha \in (n^{-1},1)$, provided $\tau \leq \frac{1}{2}\lambda_t\kappa_1^{\widetilde Z\widetilde Z}(s,3)$ holds with probability at least $1-\epsilon$, we have that
with probability $1-\alpha\{1+o(1)\}-\epsilon-o(1)$ the estimator defined in (\ref{def:nuisancej}) satisfies uniformly over $j\in S$
{\small $$ \|(\hat\mu^j,\hat\vartheta^j_0)-(\mu^j_0,\vartheta^j_0)\|_q \leq 4\tau  \frac{\|t(\mu^j_0,\vartheta^j_0)\|_\infty}{\kappa_q^{\widetilde Z\widetilde Z}(s,3)}$$}
where $\displaystyle \|t(\mu^j_0,\vartheta^j_0)\|_\infty^2= {\displaystyle\max_{\ell\in [p]\setminus\{j\}, k\in [K]}} \sum_{i=1}^n\frac{x_{i\ell}^2(z_i'\mu^j_0)^2}{n} \vee \sum_{i=1}^n\frac{x_{i\ell}^2v_{ij}^2}{n} \vee \sum_{i=1}^n\frac{z_{ik}^2v_{ij}^2}{n}$.\end{theorem}

\subsection{Linear Representation and Simultaneous Confidence Bands}

In this section we turn to the construction of confidence regions. A critical intermediary step is to establish  a linear representation for the estimators $\check\beta_j, j\in S$, despite of the high-dimensional nuisance parameters. Moreover, we will achieve results uniformly over a large class of data generating processes including many where we cannot avoid model selection mistakes.

We will assume the following condition. In what follows  the sequence $\delta_n\to 0$ is fixed, and we let $v_{ij}:= x_{ij}-z_i'\mu^j_0-x_{i,-j}'\vartheta^j_0$, and define the following moments $M_{nvz}:={\max_{j\in S, k\in [K]}}\Ep[\frac{1}{n}\sum_{i=1}^n|v^j_iz_{ik}|^3]^{1/3}/{\rm Var}(\frac{1}{\sqrt{n}}\sum_{i=1}^nv_iz_{ik})^{1/2}$,
 $M_{nxz}:={\displaystyle\max_{j\in S, k\in [p]\setminus\{j\}}}\Ep[\frac{1}{n}\sum_{i=1}^n|x_{ik}z_i'\mu_0^j|^3]^{1/3}/{\rm Var}(\frac{1}{\sqrt{n}}\sum_{i=1}^nx_{ik}z_i'\mu_0^j)^{1/2}$.

{\bf Condition IV.} {\it The data $\{W_i=(y_i,\xi_i,x_i,z_i,z_i'\mu^S_0, v_{iS}), i=1,\ldots,n\}$ is independent across $i$ and obey the model (\ref{linmodel:iv})-(\ref{linmodel:iv2}). We have $\|\beta_0\|_0+\max_{j\in S}\|\mu^j_0\|_0\leq s$ and $\|\beta_0\|_\infty+\max_{j\in S}\|\mu^j_0\|_\infty\leq C$. The following moment conditions hold: $\Ep[\frac{1}{n}\sum_{i=1}^n(W_i'\theta)^4]^{1/4} \leq C\|\theta\|$ for every $\theta \in \mathbb{R}^{{\rm dim}(W_i)}$, $\min_{j\in S} |\Ep[\frac{1}{n}\sum_{i=1}^nx_{ij}z_i'\mu^j_0]|\geq c$, $K_n = \Ep[\max_{i\leq n}\|W_i\|_\infty^q] < \infty$ for some $q>4$, $K_n^{2/q}\log^{1/2}(Kpn) \leq \delta_n n^{1/2}$, $s\log(pKn) \leq \delta_n \sqrt{n}$. With probability $1-o(1)$ we have $\kappa_q^{ZX}(s,3) \geq s^{-1/q}c$ and $\kappa_{q}^{\widetilde Z \widetilde Z}(s,3) \geq s^{-1/q}c$ for $q\in \{1,2\}$. Conditoins IV1 and IV2 also hold.}

Condition IV corresponds to standard assumptions in the literature with strong instruments. The moments conditions are similar to the ones used in \cite{BellChenChernHans:nonGauss} and allow the use of self-normalized moderate deviation theory. We note that Condition IV also imposes Conditions IV1 and IV2 which in turn allows us to use the estimators studies in sections \ref{Sec:Step1} and \ref{Sec:Step2} of the nuisance parameters. The conditions on the sensitivity quantities are comparable to the restricted eigenvalue condition in the literature of high-dimensional sparse models, see \cite{gautier:tsybakov}.

The following is among our main results. It establishes a linear representation of the proposed estimators $\check\beta_j, j\in S$.

\begin{theorem}[Uniform Linear Representation]\label{thm:main}
Under Condition IV, uniformly over $j\in S$ we have
$$\sqrt{n}\sigma_j^{-1}(\check\beta_j-\beta_{0j}) = \frac{\sigma_j^{-1}\Omega_j^{-1}}{\sqrt{n}}\sum_{i=1}^n (y_i-x_i'\beta_0)z_i'\mu^j_0 + O_\P(\delta_n).$$
\end{theorem}

An important feature of Theorem \ref{thm:main} is its validity uniformly over the class of data generating processes that satisfy  Condition IV. In particular, the result covers data generating processes for which  model selection mistakes are likely to happen. Next we build upon Theorem \ref{thm:main} to establish various useful results.

A consequence of Theorem \ref{thm:main} is the construction of confidence intervals for each component as $\sqrt{n}\sigma_j^{-1}(\check\beta_j-\beta_{0j})\rightsquigarrow N(0,1)$. Importantly, it holds uniformly over data generating processes with arbitrary small coefficients. Indeed, the orthogonality condition mitigates the impact of model selection mistakes which are unavoidable for those components.

\begin{corollary}[Componentwise Confidence Intervals]\label{cor:PointwiseInference}
Let $\mathcal{M}_n$ be the set of data generating processes that satisfies Condition IV for a fixed $n$. We have that
$$ \lim_{n\to \infty} \sup_{\mathcal{M}\in \mathcal{M}_n} \max_{j\in S} \left|\P_\mathcal{M} \left( \sqrt{n}|\check \beta_j - \beta_{0j}| \leq \Phi^{-1}(1-\alpha/2)\sigma_j \right)  - (1-\alpha)\right| = 0$$
Further, the result also holds when $\sigma_j$ is replaced by $\hat\sigma_j$ as defined in (\ref{est:sigmaj}).
\end{corollary}

Next we turn to simultaneous confidence bands over $S\subset\{1,\ldots,p\}$ components of $\beta_0$. We allow for the cardinality of $S$ to also grow with the sample size (and potentially $S=\{1,\dots,p\}$). This exploits new central limit theorems for high dimensional vectors, see \cite{chernozhukov2015noncenteredprocesses} and the references therein. The following result provides sufficient conditions under which the multiplier bootstrap procedure based  on (\ref{def:mb}) yields (honest) simultaneous confidence bands.

In what follows we let $A_n\geq \max_{j\in S}\Ep[\frac{n^{-1}}{\sigma_j^{3}\Omega_j^{3}}\sum_{i=1}^n|\xi_iz_i'\mu^j_0|^3]$ and $B_{\bar q,n} \geq \frac{1}{n}\sum_{i=1}^n\Ep[ \max_{j\in S}|\xi_iz_i'\mu^j_0|^{\bar q}]$ for some $\bar q>3$. Note that $B_{\bar q,n} \leq K_n^{\bar q}$ provided that $\bar q \leq q/2$. However, $B_{\bar q,n}$ can be bounded in various settings if $|S|$ is bounded.

\begin{theorem}\label{thm:inference}
Let $\mathcal{M}_n$ be the set of data generating processes that satisfy Condition IV for a fixed $n$. Furthermore, suppose that
$\delta_n\log^{3/2}(n|S|)(1+\max_{j\in S}\|\mu^j_0\|)=o(1)$, $A_n^2\log^7(|S|)=o(n)$, $B_{\bar q,n}^{1/\bar q}\log^{2}(|S|)=o(n^{1/2-1/\bar q})$, and $K_n^{4/q}\log^{5}(|S|)=o(n)$ holds.
For the critical value $c_{\alpha,S}$ computed via the multiplier bootstrap procedure, we have that
$$ \lim_{n\to \infty} \sup_{\mathcal{M}\in \mathcal{M}_n} \left|\P_\mathcal{M} \left( \check \beta_j - \frac{c_{\alpha,S}^* \hat\sigma_j}{\sqrt{n}}\leq \beta_{0j} \leq \check \beta_j + \frac{c_{\alpha,S}^* \hat \sigma_j}{\sqrt{n}}, \ \ \forall j\in S\right)  - (1-\alpha)\right| = 0$$
\end{theorem}

Theorem \ref{thm:inference} yields simultaneous confidence regions that are uniformly valid across models that satisfy Condition IV and the additional growth conditions on the $n$, $\log |S|$ and various moments. In particular, they allow for (sequence of) models where model selection mistakes are unavoidable.

\section{Monte Carlo Experiments}

In this section we investigate the finite sample performance of the proposed estimators via numerical simulations. We consider the following data generating process
$$ y = x'\beta_0 + \varepsilon, \ \ \ x_j = \tilde x_j + \sum_{\ell=1}^Lz_{L(j-1)+k}, \ \ \ \varepsilon = \zeta + \tilde x'\gamma_0 $$
where $\beta_0=(1,.8,.6,.4,.2,0,\ldots)'$, $\gamma_0=(.1,.2,.3,.4,.5,.6,.7,.8,.9,1,0,\ldots)'$, $z \sim N(0,I_{Fp})$, $\tilde x \sim N(0,\Sigma)$, $\Sigma_{ij}=0.3^{|i-j|}$, and $\zeta \sim N(0,1/4^2)$. This corresponds to the case that $L$ instruments are useful for each covariate.

\begin{table}[h!]
\begin{tabular}{cccc|cc|ccc}
    & &  &  &   \multicolumn{2}{c|}{Uniform over $\{1,2,3\}$} & \multicolumn{3}{c}{Bias}\\
$n$ & $p$ & $K$ & $L$ & rp(.05) & $\ell_\infty$ & $\beta_{01}$ & $\beta_{02}$ & $\beta_{03}$  \\
\hline
500 & 30 &  30 &  1 & 0.060 & 0.1744 & -0.0130 & -0.0244 & -0.0199 \\
500 & 30 &  90 &  3 & 0.053 & 0.0868 & -0.0019 &  0.0022 & 0.0010 \\
500 & 100 & 100 & 1 & 0.040 & 0.1636 & -0.0091 & -0.0108& -0.0199 \\
500 & 100 & 300 & 3 & 0.077 & 0.0831 & -0.0011 & -0.0030& 0.0013 \\
500 & 300 & 300 & 1 & 0.057 & 0.1451 & -0.0056 & -0.0089& -0.0183 \\
500 & 300 & 900 & 3 & 0.060 & 0.0734 & 0.0012 & 0.0012&  0.0029\\
\end{tabular}
\end{table}

The table above illustrates how the proposed procedure accurately controls size where the confidence region aims to simultaneously cover the first three components. The table also reports the bias in the estimation of each component. To achieve correct size in this high-dimensional settings it is critical to obtain a low bias and this is what the table reports.

{\bf Acknowledgement.} We thank Victor Chernozhukov and Denis Chetverikov for several comments and suggestions.

\begin{appendix}

\section{Proofs of Section 3}

\begin{proof}[Proof of Theorem \ref{thm:ExtraSensitivity}]  Let $T \subset [p]$ and $\delta \in C_T(u)$, i.e., $\|\delta_{T^c}\|_1 \leq u\|\delta_T\|_1$.
Let $m\in \mathbb{N}$ and denote by $T_1$ the support of the $m$-largest components (in absolute value) of $\delta$, $T_2$ the subsequent $m$-largest, etc. For any $\mu \in \mathbb{R}^K$, such that $\|\mu\|_0\leq m$, we have
$$\begin{array}{rl}
\|\Psi \delta \|_\infty &  = \max_{\mu } \left|\frac{\mu'}{\|\mu\|_1}\Psi \delta\right| \geq  \left|\frac{\mu'}{\sqrt{m}\|\mu\|}\Psi \delta\right| \geq \left|\frac{\mu'}{\sqrt{m}\|\mu\|}\Psi \delta_{T_1}\right| - \left|\frac{\mu'}{\sqrt{m}\|\mu\|}\Psi \delta_{T_1^c}\right| \\
& \geq \left|\frac{\mu'}{\sqrt{m}\|\mu\|}\Psi \delta_{T_1}\right| - \sum_{\ell\geq 2}\left|\frac{\mu'}{\sqrt{m}\|\mu\|}\Psi \delta_{T_\ell}\right| \\
& \geq \left|\frac{\mu'}{\sqrt{m}\|\mu\|}U\Lambda V \delta_{T_1}\right| - \sum_{\ell\geq 2} \left | \frac{\mu'}{\sqrt{m}\|\mu\|} \Psi  \delta_{T_\ell}\right | \\
\end{array}
$$
where $U\Lambda V$ denotes the SVD of $\Psi_{J_\mu,T_1}$. By definition of $\sigma_{\min}(m)$, by suitably choosing $J_{\mu} = \text{support}(\mu)$ to correspond to the largest singular values $\Lambda_{jj}$, $j=1,...,m$, we have that $\Lambda_{jj} \geq \sigma_{\min}(m)$ for $j=1,\ldots,m$. We can assume $\sigma_{\min}(m)>0$ otherwise the result is trivial.

Since the columns of $U'$ associated to $J_\mu$ span $\mathbb{R}^m$, we can let $\mu$ be such that $U'\mu=V\delta_{T_1}$, so that $\|\mu\|=\|\delta_{T_1}\|$, and recall that by assumption $\Lambda_{jj} \geq \sigma_{\min}(m)$. Hence
$$
\left|\frac{\mu'}{\sqrt{m}\|\mu\|}U\Lambda V \delta_{T_1}\right| \geq  \frac{\sigma_{\min}(m)}{\sqrt{m}}\|\delta_{T_1}\|.
$$
Moreover, we have
$$
\left|\frac{\mu'}{\sqrt{m}\|\mu\|} \Psi  \delta_{T_\ell}\right|=\left|\frac{\mu'}{\sqrt{m}\|\mu\|} \Psi_{J_\mu,T_\ell}  \delta_{T_\ell}\right|\leq \frac{\sigma_{\max}(m)}{\sqrt{m}} \| \delta_{T_\ell}\|.
$$ where the last inequality follows by definition of $\sigma_{\max}$ and $|J_\mu|\leq m$, $|T_\ell|\leq m$.

By construction of the sets $\{T_\ell\}_{\ell\geq 1}$, we have $\|\delta_{T_{\ell}}\|_1/\sqrt{m} \geq \|\delta_{T_{\ell+1}}\|$.

Thus we have
$$\begin{array}{rl}
\|\Psi \delta \|_\infty &
 \geq \frac{\sigma_{\min}(m)}{\sqrt{m}}\|\delta_{T_1}\| - \frac{\sigma_{\max}(m)}{\sqrt{m}}\sum_{\ell\geq 1}\|\delta_{T_\ell}\|_1/\sqrt{m} \\
& \geq \frac{\sigma_{\min}(m)}{\sqrt{m}}\|\delta_{T_1}\| - \frac{\sigma_{\max}(m)}{\sqrt{m}}(1+u)\|\delta_{T}\|_1/\sqrt{m} \\
& \geq \frac{\sigma_{\min}(m)}{\sqrt{m}}\|\delta_{T_1}\| - \frac{\sigma_{\max}(m)}{\sqrt{m}}(1+u)\|\delta_{T}\|\sqrt{s/m} \\
\end{array}
$$
where we used that 
$\|\delta\|_1 \leq (1+u)\|\delta_T\|_1$,
which is implied by $\| \theta_0 + \delta\|_1 \leq \| \theta_0\|_1$,  and $|T|=s$.

To conclude the proof note that by a similar argument $$\|\delta\| \leq \|\delta_{T_1}\| + \sqrt{s/m}(1+u)\|\delta_T\| \leq \{1+(1+u)\sqrt{s/m}\}\|\delta_{T_1}\|$$
and that $\|\delta\|/\|\delta\|_1 \geq 1/\{\sqrt{s}(1+u)\}$. \end{proof}

\begin{proof}[Proof of Corollary \ref{cor:ExtraSensitivity-WeakIV}]

Take $\theta \in C_J(u)$ for some $J \subset [p]$, $|J|\leq s$. Note that
$$\begin{array}{rl}
\frac{\|\Psi \theta \|_\infty}{\|\theta\|_q}
& \geq_{(1)} {\displaystyle \max_{m \geq s }} \left\{\frac{\sigma_{\min}(m)}{\sqrt{m}} - \frac{\sigma_{\max}(m)}{\sqrt{m}}(1+u)\sqrt{s/m}\right\}\frac{s^{1/2-1/q}}{\{1+(1+u)\sqrt{s/m}\}(1+u)}  \\
& \geq_{(2)} \frac{s^{1/2-1/q}}{8} \frac{\mu_n^2}{16\sqrt{s}}  \\
  & = \frac{s^{-1/q}}{8} \frac{\mu_n^2}{16}  \\
\end{array}
$$ where (1) follows from Theorem \ref{thm:ExtraSensitivity}, (2) by setting $m=64s/\mu_n^2$ and $u=3$.
\end{proof}

\begin{proof}[Proof of Corollary \ref{cor:ExtraSensitivity}]

Take $\theta \in C_J(u)$ for some $J \subset [p]$, $|J|\leq s$. Note that

$$\begin{array}{rl}
\frac{\|\Psi \theta \|_\infty}{\|\theta\|_q} & \geq_{(1)} \| \Ep[zx'] \theta\|_\infty/\|\theta\|_q - \| (\En-\Ep)[zx']\theta\|_\infty/\|\theta\|_q \\
& \geq_{(2)} {\displaystyle \max_{m \geq s }} \left\{\frac{\sigma_{\min}(m)}{\sqrt{m}} - \frac{\sigma_{\max}(m)}{\sqrt{m}}(1+u)\sqrt{s/m}\right\}\frac{s^{1/2-1/q}}{\{1+(1+u)\sqrt{s/m}\}(1+u)}  \\
& - \| (\En-\Ep)[zx']\theta\|_\infty/\|\theta\|_q \\
& \geq_{(3)} \frac{s^{1/2-1/q}}{2(1+u)} {\displaystyle \max_{m \geq s (1+u)^2 }} \left\{\frac{c}{\sqrt{m}} - \frac{C(1+u)\sqrt{s/m}}{\sqrt{m}}\right\} \\
   &- \| (\En-\Ep)[zx']\|_\infty\|\theta\|_1/\|\theta\|_q \\
& \geq_{(4)} \frac{s^{-1/q}}{2(1+u)} \frac{c}{2} \frac{c}{2C(1+u)}   - \| (\En-\Ep)[zx']\|_\infty\|\theta\|_1/\|\theta\|_q \\
\end{array}
$$ where (1) follows from the triangle inequality, (2) follows from Theorem \ref{thm:ExtraSensitivity}, (3) by restricting the maximum and Holder's inequality, and (4) by setting $m=4sC^2(1+u)^2/c^2$ and using that $s(1+u)^2\log n \geq m$ for $n$ large.

Next we control $\| (\En-\Ep)[zx']\|_\infty$ via Lemma \ref{lem:m2bound}(2). Note that $|\Ep[z_{ij}x_{ik}]|\leq \frac{1}{2}\Ep[z_{ij}^2+x_{ik}^2] \leq K_n$. Therefore
$$ \Ep[ \| (\En-\Ep)[zx']\|_\infty] \leq C' \sqrt{ K_n n^{-1} \log(Kp)} + K_n^{1/2} n^{-1}\log(Kp)$$
Then with probability $1-\ell_n$, for some $\ell_n\to 0$, by Markov's inequality
$$  \| (\En-\Ep)[zx']\|_\infty \leq C'' \sqrt{ K_n n^{-1} \log(Kp)}/\ell_n$$
Moreover, if $q=1$, $\|\theta\|_1/\|\theta\|_q =1$, and if $q=2$, $\|\theta\|_1/\|\theta\|_q \leq (1+u)s^{1/2}$.

Therefore, with probability $1-\ell_n$ we have
$$ \| (\En-\Ep)[zx']\|_\infty\|\theta\|_1/\|\theta\|_q  \leq  C''' (1+u)s^{1-1/q} \sqrt{ K_n n^{-1} \log(Kp)}/\ell_n $$

The result follows provided $n$ is large enough so that
$$ C''' (1+u)s^{1-1/q} \sqrt{ K_n n^{-1} \log(Kp)}/\ell_n \leq \frac{s^{-1/q}}{4(1+u)} \frac{c}{2} \frac{c}{2C(1+u)}$$
which is possible for $\ell_n\to 0$ slowly enough under $(1+u)^6s^2 K_n n^{-1}\log(Kp) = o(1)$, and we obtain
$$ \kappa_q^{\Psi}(s)  \geq s^{-1/q} c  \ \frac{c/C}{16(1+u)^2} \ \ q\in\{1,2\}$$
\end{proof}

\begin{proof}[Proof of Theorem \ref{thm:linmodel:iv}]
Let $t_j(\beta)=\{\frac{1}{n}\sum_{i=1}^nz_{ij}^2(y_i-x_i'\beta)^2\}^{1/2}$ and $U_{ij} = z_{ij}(y_i-x_i'\beta_0) = z_{ij}\xi_i$ is such that $M_{nz\xi} \geq \Ep[\frac{1}{n}\sum_{i=1}^n|U_{ij}|^3]^{1/3} /\Ep[\frac{1}{n}\sum_{i=1}^n|U_{ij}|^2]^{1/2}$ for all $j \leq K$. By Lemma 7.4 in \cite{delapena}, we have
\begin{equation}\label{linmodel:eqsn}\begin{array}{rl}
 \P\left( \left| \frac{1}{n}\sum_{i=1}^n z_{ij}(y_i-x_i'\beta_0) \right|
>  \tau t_j(\beta_0)  \right)
& = \P\left( \left| \frac{1}{\sqrt{n}}\sum_{i=1}^n U_{ij} \right|
> \sqrt{n}\tau \left\|\frac{U_{\cdot j}}{\sqrt{n}}\right\|  \right) \\
& \leq \{1-\Phi(\sqrt{n}\tau)\}\left( 1 + \frac{A}{\ell_n^3} \right)\\
&=\frac{\alpha}{2K}\left( 1 + \frac{A}{\ell_n^3} \right) \\
& \end{array}
\end{equation}
for some $\ell_n\to \infty$ given the assumption $M_{nz\xi}\log^{1/2}(Knp)=o(n^{1/6})$ in Condition IV1 and the choice of $\tau$. Then by the union bound we have that $\beta_0$ is feasible in the proposed optimization problem with probability at least $1-\alpha\left( 1 + \frac{A}{\ell_n^3} \right)$. Therefore
\begin{equation}\label{linmodel:opt:iv} \|\hat\beta\|_1 + \lambda_t \|\hat t\|_\infty \leq \|\beta_0\|_1 + \lambda_t\|t(\beta_0)\|_\infty\end{equation}
Since $\|t(\hat\beta)\|_\infty\leq \|\hat t\|_\infty$ and $  |t_j(\hat\beta) -  t_j(\beta_0)|\leq  H_n\|\hat\beta-\beta_0\|_1$, by
setting $\lambda_t = \frac{1}{2H_n}$  we have
$$ \frac{1}{2}\|\hat \beta_{T^c}\|_1 \leq \frac{3}{2}\|\beta_0-\hat\beta_T\|_1$$

Next note that with probability $1-\alpha(1+o(1))$
$$\begin{array}{rl}
\|\frac{1}{n}Z'X(\hat\beta-\beta_0)\|_\infty & \leq \|\frac{1}{n}Z'(Y-X\hat\beta)\|_\infty + \|\frac{1}{n}Z'\xi\|_\infty\\
& \leq \tau\|\hat t\|_\infty + \tau\|t(\beta_0)\|_\infty\\
& \leq \frac{\tau}{\lambda_t}\|\hat \beta-\beta_0\|_1 + 2\tau\|t(\beta_0)\|_\infty\\
& \leq \frac{\tau}{\lambda_t\kappa_1^{ZX}(s,3)}\|\frac{1}{n}Z'X(\hat\beta-\beta_0)\|_\infty + 2\tau\|t(\beta_0)\|_\infty\\
\end{array}$$
where the first inequality follows from the triangle inequality, the second by feasibility of $(\hat\beta,\hat t)$ and (\ref{linmodel:eqsn}), and the third inequality follows from (\ref{linmodel:opt:iv}). The fourth inequality follows from the definition of the $\ell_1$-sensitivity $\kappa_1^{ZX}(s,3)$.

Therefore we have
$$
\begin{array}{rl}
\left(1-\frac{\tau}{\lambda_t\kappa_1^{ZX}(s,3)}\right)\kappa_q^{ZX}(s,3)\|\hat\beta-\beta_0\|_q \leq 2\tau\|t(\beta_0)\|_\infty
\end{array}
$$
and the result follows from the condition $\tau \leq 2^{-1}\lambda_t\kappa_1^{ZX}(s,3)$ holding with probability $1-\epsilon$.
\end{proof}

\begin{proof}[Proof of Theorem \ref{thm:nuisancej}]
We begin by defining some convenient notation. We fix a $j\in S$. Let $t=(t^x{}',t^z{}',t^{xz}{}')' $ be a vector in $\mathbb{R}^{2(p-1)+(K-1)}$. Define the functions
$$\begin{array}{rl}
t_{\ell}^{xz}(\mu)=\{\frac{1}{n}\sum_{i=1}^n|x_{i\ell}z_i'\mu|^2\}^{1/2}, \ \ell \in [p]\setminus \{j\}\\
t_{\ell}^x(\mu,\vartheta)=\{\frac{1}{n}\sum_{i=1}^n|x_{i\ell}(x_j-z_i'\mu-x_{i,-j}'\vartheta)|^2\}^{1/2}, \ell \in [p]\setminus \{j\}\\ t_{\ell}^z(\mu,\vartheta)=\{\frac{1}{n}\sum_{i=1}^n|z_{i\ell}(x_j-z_i'\mu-x_{i,-j}'\vartheta)|^2\}^{1/2}, k\in [K].\end{array}$$ Next define the random variables $\widetilde U_{i\ell}^{xz} = x_{i\ell}z_i'\mu_0^j$, $\widetilde U_{i\ell}^x=x_{i\ell}v_{ij}$ and $\widetilde U_{ik}^z=z_{i\ell}v_{ij}$, $\ell \in [p]\setminus\{j\}$ and $k \in [K]$, so that $\Ep[\frac{1}{n}\sum_{i=1}^n\widetilde U_{i\ell}^{xz}]=0$, $\Ep[\frac{1}{n}\sum_{i=1}^n\widetilde U_{i\ell}^{x}]=0$, and $\Ep[\frac{1}{n}\sum_{i=1}^n\widetilde U_{ik}^{z}]=0$ by (\ref{ortho:new}) and the definition of $(\mu_0^j,\vartheta_0^j)$. Finally we define $U_{i\ell}^h = \widetilde U_{i\ell}^h - \Ep[U_{i\ell}^h]$ for $h=z,x,zx$. Therefore we have $M_{h,j} \geq \Ep[\frac{1}{n}\sum_{i=1}^n|U_{i\ell}^h|^3]^{1/3}/\Ep[\frac{1}{n}\sum_{i=1}^n|U_{i\ell}^h|^2]^{1/2}$ for $h=x,z,xz$ by Condition IV2.

Since $\sum_{i=1}^n U_{i\ell}^h=\sum_{i=1}^n \widetilde U_{i\ell}^h$ since $\Ep[\sum_{i=1}^n \widetilde U_{i\ell}^h]=0$ we have that
\begin{equation}\label{linmodel:eqsn20}\begin{array}{ll}
\P\left(\exists (\ell,h): \left| \frac{1}{\sqrt{n}}\sum_{i=1}^n \widetilde U_{i\ell}^h \right|
> c\sqrt{n}\tau \left\|\frac{\widetilde U_{\cdot \ell}^h}{\sqrt{n}}\right\|  \right)
\\ = \P\left(\exists (\ell,h):  \left| \frac{1}{\sqrt{n}}\sum_{i=1}^n  U_{i\ell}^h \right|
> c\sqrt{n}\tau \left\|\frac{\widetilde U_{\cdot \ell}^h}{\sqrt{n}}\right\|  \right)\\
\leq \P\left(\exists (\ell,h):  \left| \frac{1}{\sqrt{n}}\sum_{i=1}^n  U_{i\ell}^h \right|
> \sqrt{n}\tau \left\|\frac{U_{\cdot \ell}^h}{\sqrt{n}}\right\|  \right)\\
+ \P\left(\exists (\ell,h): \left\|U_{\cdot \ell}^h\right\| > c\left\|\widetilde U_{\cdot \ell}^h\right\|  \right)\\
\end{array}
\end{equation} where the constant $c$ was taken as $c=1.1$.

To bound the first term in the RHS of (\ref{linmodel:eqsn20}) note that $\Ep[U_{i\ell}^h]=0$ for $h=z,x,{xz}$ and each corresponding $\ell$. Then by Lemma 7.4 in \cite{delapena}, for $h=x,z,xz$ and each corresponding $\ell$, we have
\begin{equation}\label{linmodel:eqsn2}\begin{array}{rl}
\P\left( \left| \frac{1}{\sqrt{n}}\sum_{i=1}^n U_{i\ell}^h \right|
> \sqrt{n}\tau \left\|\frac{U_{\cdot \ell}^h}{\sqrt{n}}\right\|  \right)
& \leq \{1-\Phi(\sqrt{n}\tau)\}\left( 1 + \frac{A}{\ell_n^3} \right)\\
&=\frac{\alpha}{2|S|(K+2p)}\left( 1 + \frac{A}{\ell_n^3} \right) \\
& \end{array}
\end{equation}
for some $\ell_n\to \infty$ given our condition $M_{h,j}\log^{1/2}(|S|Kpn)=o(n^{1/3})$, $h=z,x$, and the choice of $\tau$.

To bound the second term in the RHS of (\ref{linmodel:eqsn20}) we have
\begin{equation}\label{linmodel:eqsn3}\begin{array}{ll}
 & \P\left(\exists (\ell,h): \sum_{i=1}^n(U_{i\ell}^h)^2 > c^2\sum_{i=1}^n(U_{i\ell}^h+\Ep[\widetilde U_{i\ell}^h])^2  \right)\\
  & \leq \P\left(\exists (\ell,h): \sum_{i=1}^nU_{i\ell}^h\Ep[\widetilde U_{i\ell}^h]  > \frac{c^2-1}{2c^2}\sum_{i=1}^n(U_{i\ell}^h)^2+(\Ep[\widetilde U_{i\ell}^h])^2  \right)\\
  & \leq \P\left(\exists (\ell,h): \sum_{i=1}^nU_{i\ell}^h\Ep[\widetilde U_{i\ell}^h]  \geq c'n\right)\\
   & + \P\left(\exists (\ell,h): \frac{c^2-1}{2c^2}\sum_{i=1}^n(U_{i\ell}^h)^2+(\Ep[\widetilde U_{i\ell}^h])^2  < c'n \right)\\
   & = o(1)
  \end{array}
\end{equation}
Indeed we have that for $A_n = \frac{1}{n}\sum_{i=1}^n\Ep[(U_{i\ell}^h)^21\{|U_{i\ell}^h|\leq \sqrt{\log n}\}]$
$$\begin{array}{rl}
\P\left(\sum_{i=1}^n(U_{i\ell}^h)^2 < A_n n - tn \right) & \leq \P\left(\sum_{i=1}^n(U_{i\ell}^h)^21\{|U_{i\ell}^h|\leq \sqrt{\log n}\} < A_n n -tn \right)\\
 &\leq \exp(-t^2n/2\log^2 n)\\
\end{array}
$$
by Hoefding inequality. Note that $\Ep[\frac{1}{n}\sum_{i=1}^n(U_{i\ell}^h)^2]-A_n \leq \Ep[\sum_{i=1}^n(U_{i\ell}^h)^21\{|U_{i\ell}^h|> \sqrt{\log n}\}] \leq \Ep[ \frac{1}{n}\sum_{i=1}^n|U_{i\ell}^h|^3 ]/\sqrt{\log n} = o(1)$ by Condition IV2. Also by Condition IV2 we have $\Ep[\frac{1}{n}\sum_{i=1}^n(U_{i\ell}^h)^2]\geq \tilde c$.

Next by Markov's inequality
$$\P\left(\exists (\ell,h): \sum_{i=1}^nU_{i\ell}^h\Ep[\widetilde U_{i\ell}^h]  \geq c'n\right) \leq \Ep\left[\max_{\ell, h} \left|\frac{1}{n}\sum_{i=1}^nU_{i\ell}^h\Ep[\widetilde U_{i\ell}^h]\right|\right]/c'$$
Then by Lemma \ref{lem:m2bound}(2) we have
$$\begin{array}{rl}
 \P\left(\exists (\ell,h): \sum_{i=1}^nU_{i\ell}^h\Ep[\widetilde U_{i\ell}^h]  \geq c'n\right) & \leq C\sqrt{n^{-1}\log(Kpn)} \\
 & + Cn^{-1}\Ep[\max_{i,\ell,h}|U_{i\ell}^h\Ep[\widetilde U_{i\ell}^h]|^2]^{1/2}\log(Kpn)\\
 & = o(1) \end{array}$$
since $Cn^{-1}\Ep[\max_{i,\ell,h}|U_{i\ell}^h\Ep[\widetilde U_{i\ell}^h]|^2]^{1/2}\log(Kpn)=o(1)$ by Condition IV2.

Then by combining the union bound with (\ref{linmodel:eqsn2}) and (\ref{linmodel:eqsn3}), we have that $(\mu_0^j,\vartheta^j_0)$ is feasible in the proposed optimization problem with probability at least $1-\alpha\left( 1 + \frac{A}{\ell_n^3} \right)-o(1)$. Therefore
\begin{equation}\label{linmodel:opt:iv} \|\hat\mu^j\|_1+\|\hat\vartheta^j\|_1 + \lambda_t \|\hat t\|_\infty \leq \|\mu_0^j\|_1 + \|\vartheta^j_0\|_1+\lambda_t\|t(\mu_0^j,\vartheta^j_0)\|_\infty\end{equation}

Next note that
\begin{equation}\label{auxboundt}
\begin{array}{rl}
|t_k^{xz}(\hat\mu^j) -  t_k^{xz}(\mu_0^j)|^2 & = | \{\frac{1}{n}\sum_{i=1}^n|x_{ik}z_i'\hat\mu^j|^2\}^{1/2} - \{\frac{1}{n}\sum_{i=1}^n|x_{ik}z_i'\mu_0^j|^2\}^{1/2} |^2 \\
& \leq \frac{1}{n}\sum_{i=1}^n|x_{ik}z_i'(\hat\mu^j-\mu^j_0)|^2\\
& \leq \|\hat\mu^j-\hat\mu^j_0\|_1^2 H_n^2\\
|t_k^z(\hat\mu^j,\hat\vartheta^j) -  t_k^z(\mu_0^j,\vartheta^j_0)|^2 & \leq \frac{1}{n}\sum_{i=1}^n|z_{ik}\{z_i'(\hat\mu^j-\mu^j_0)+x_{i,-j}'(\hat\vartheta^j-\vartheta^j_0)\}|^2\\
& \leq \{\|\hat\mu^j-\hat\mu^j_0\|_1+\|\hat\vartheta^j-\vartheta^j_0\|_1\}^2 H_n^2\\
|t_k^x(\hat\mu^j,\hat\vartheta^j) -  t_k^x(\mu_0^j,\vartheta^j_0)|^2 & \leq \frac{1}{n}\sum_{i=1}^n|x_{ik}\{z_i'(\hat\mu^j-\mu^j_0)+x_{i,-j}'(\hat\vartheta^j-\vartheta^j_0)\}|^2\\
& \leq \{\|\hat\mu^j-\hat\mu^j_0\|_1+\|\hat\vartheta^j-\vartheta^j_0\|_1\}^2 H_n^2
\end{array}
\end{equation}
where $H_n \geq {\displaystyle \max_{k\in[K],\ell\in[p]}}\left\{\|\En[z_{k}^2zz']\|_\infty, \|\En[z_{k}^2xx']\|_\infty, \|\En[x_{\ell}^2zz']\|_\infty\right\}$.

Next let $T_j := \supp(\mu_j^0,\vartheta^j_0)$. Since $\|t(\hat\mu^j,\hat\vartheta^j)\|_\infty\leq \|\hat t\|_\infty$,  using (\ref{auxboundt}) and the definition of  $\lambda_t = \frac{1}{2H_n}$, by (\ref{linmodel:opt:iv}) we have
$$\begin{array}{rl}
\|(\hat\mu^j,\hat\vartheta^j)_{T_j^c}\|_1 & \leq \|(\mu^j_0,\vartheta^j_0)\|_1 + \lambda_t\|t(\mu^j_0,\vartheta^j_0)\|_\infty -  \lambda_t \|t(\hat\mu^j,\hat\vartheta^j)\|_\infty-\|(\hat\mu^j,\hat\vartheta^j)_{T_j}\|_1\\
& \leq \|(\mu_0^j, \vartheta^j_0)-(\hat\mu^j,\hat\vartheta^j)_{T_j}\|_1 + \lambda_t\|t(\mu^j_0,\vartheta^j_0)-t(\hat\mu^j,\hat\vartheta^j)\|_\infty\\
& \leq \|(\mu_0^j, \vartheta^j_0)-(\hat\mu^j,\hat\vartheta^j)_{T_j}\|_1 + \frac{1}{2}\|(\mu_0^j, \vartheta^j_0)-(\hat\mu^j,\hat\vartheta^j)\|_1\\
\end{array}$$
so that $\|(\hat\mu^j,\hat\vartheta^j)_{T_j^c}\|_1\leq 3\|(\mu_0^j, \vartheta^j_0)-(\hat\mu^j,\hat\vartheta^j)_{T_j}\|_1$. Furthermore, note that
\begin{equation}\label{auxboundt2} |\hat t_\ell| \leq  \|t(\mu^j_0,\vartheta^j_0)\|_\infty + \frac{\|(\mu_0^j, \vartheta^j_0)-(\hat\mu^j,\hat\vartheta^j)\|_1}{\lambda_t} \end{equation}

Next note that with probability $1-\alpha\{1+o(1)\}$ for all $j\in S$ we have $(\mu^j_0,\vartheta^j_0)$ is feasible to the corresponding (\ref{def:nuisancej}). Then defining $\widetilde Z = [ Z ; X_{-j} ]$, $\theta^j_0 = (\mu^j_0,\vartheta^j_0)$,  and $\hat \theta^j = (\hat\mu^j,\hat\vartheta^j)$, we have
$$\begin{array}{rl}
\|\frac{1}{n}\widetilde Z'\widetilde Z(\hat\theta^j-\theta_0^j)\|_\infty & \leq \|\frac{1}{n}\widetilde Z'(X_j-\widetilde Z\hat\theta^j)\|_\infty + \|\frac{1}{n}\widetilde Z'(X_j-\widetilde Z\theta^j_0)\|_\infty\\
& \leq \tau\|\hat t^x\|_\infty + \tau\|t(\mu^j_0,\vartheta^j_0)\|_\infty\\
& \leq \frac{\tau}{\lambda_t}\|\hat\theta^j-\theta_0^j\|_1 + 2\tau\|t(\mu^j_0,\vartheta^j_0)\|_\infty\\
& \leq \frac{\tau}{\lambda_t\kappa_1^{\widetilde Z\widetilde Z}(s,3)}\|\frac{1}{n}\widetilde Z'\widetilde Z(\hat\theta^j-\theta_0^j)\|_\infty + 2\tau\|t(\mu^j_0,\vartheta^j_0)\|_\infty\\
\end{array}$$
where the first inequality follows from the triangle inequality, the second by feasibility of $(\hat\mu^j,\hat t)$ and (\ref{linmodel:eqsn}), and the third inequality follows from (\ref{auxboundt2}). The fourth inequality follows from the definition of the $\ell_1$-sensitivity $\kappa_1^{\widetilde Z\widetilde Z}(s,3)$.

Therefore we have
$$
\begin{array}{rl}
\left(1-\frac{\tau}{\lambda_t\kappa_1^{\widetilde Z\widetilde Z}(s,3)}\right)\kappa_q^{\widetilde Z\widetilde Z}(s,3)\|\hat\theta^j-\theta_0^j\|_q \leq 2\tau\|t(\mu^j_0,\vartheta^j_0)\|_\infty
\end{array}
$$
provided the condition $\tau \leq 2^{-1}\lambda_t\kappa_1^{\widetilde Z\widetilde Z}(s,3)$ holding with probability $1-\epsilon$.
\end{proof}

\begin{proof}[Proof of Theorem \ref{thm:main}]
For convenience we will use the notation
$\En[\cdot]=\frac{1}{n}\sum_{i=1}^n[\cdot_i]$. By (\ref{def:checkbeta}) we have
\begin{equation}\label{def:newbetaj} \check\beta_j := \hat \Omega_j^{-1} \En[(y-x'\hat\beta_{-j})z'\hat \mu^j]\end{equation}
where $\hat \Omega_j := \En[ x_{j}z'\hat\mu^j]$.
Next we rearrange the expression (\ref{def:newbetaj}). It follows that

{\small $$\begin{array}{rl}
  \check\beta_j & = \hat \Omega_j^{-1}\En[\{\xi+x_{j}\beta_{0j}-x'(\hat\beta_{-j}-\beta_{0,-j})\}z'\hat\mu^j]\\
  & = \hat \Omega_j^{-1}\beta_{0j}\En[x_{j}z'\hat\mu^j]+\hat \Omega_j^{-1}\En[\{\xi-x_{-j}'(\hat\beta_{-j}-\beta_{0,-j})\}z'\hat\mu^j]\\
  & = \beta_{0j} + \hat \Omega_j^{-1}\En[\xi z'\hat\mu^j] - \hat \Omega_j^{-1}(\hat\beta_{-j}-\beta_{0,-j})'\En[x_{i,-j}z_i'\hat\mu^j]\\
  & = \beta_{0j} + \hat \Omega_j^{-1}\En[\xi z'\mu^j_0] + \hat \Omega_j^{-1}\En[\xi z'](\hat\mu^j-\mu^j_0)\\
  &  - \hat \Omega_j^{-1}(\hat\beta_{-j}-\beta_{0,-j})'\En[x_{-j}z'\mu^j_0] - \hat \Omega_j^{-1}(\hat\beta_{-j}-\beta_{0,-j})'\En[x_{-j}z'](\hat\mu^j-\mu^j_0) \\
  & =: \beta_{0j} + \hat \Omega_j^{-1}\En[\xi z'\mu^j_0] + B_1^j - B_2^j - B_3^j
 \end{array}$$}

Therefore we have
$$\begin{array}{rl}
\sqrt{n}\hat\Omega_j(\check\beta_j-\beta_{0j}) & = \sqrt{n}\En[\xi z'\mu^j_0]  + O_\P(\delta_n) \\
\end{array} $$

Step 2 below shows that with probability $1-o(1)$ uniformly over $j\in S$
$$|B_1^j| + |B_2^j| + |B_3^j| \leq C s \log (pKn) / n.  $$

Under Condition IV, we have that $s \log (pKn) \leq \delta_n \sqrt{n}$, so that we obtain the stated linear representation, namely uniformly over $j\in S$, with probability $1-o(1)$ we have
$$\begin{array}{rl}
\sqrt{n}\hat\Omega_j(\check\beta_j-\beta_{0j}) & = \frac{1}{\sqrt{n}}\sum_{i=1}^n\psi_j(y_i,x_i,z_i)  + O_\P(\delta_n) \\
\end{array} $$
where $\psi_j(y_i,x_i,z_i)=(y_i-x_i'\beta_0)z_i'\mu^j_0$ is a zero mean random variable.

The result follows by (\ref{boundOmegaj}) noting that
$$\begin{array}{rl}
\sqrt{n}|(\hat\Omega_j-\Omega_j)(\check\beta_j-\beta_{0j})| & \leq \max_{j\in S}|\hat\Omega_j-\Omega_j|\frac{\hat\Omega_j^{-1}}{\sqrt{n}}\sum_{i=1}^n\psi_j(y_i,x_i,z_i) \\
& + |\hat\Omega_j-\Omega_j|O_\P(\delta_n)\\
& \leq O_\P( \delta_n \log^{-1/2}(Kpn) ) {\displaystyle \max_{j\in S}\frac{1}{\sqrt{n}}\sum_{i=1}^n\psi_j(y_i,x_i,z_i)} \\
& \leq O_\P( \delta_n ) \end{array}$$
since $\hat\Omega_j$ is bounded away from zero with probability $1-o(1)$ by Step 2 below, $\max_{j\in S}\frac{1}{\sqrt{n}}\sum_{i=1}^n\psi_j(y_i,x_i,z_i) = O_\P ( \log^{1/2}|S|)$ and $|S|\leq p$.

Step 2. (Auxiliary Calculations for $B_1^j, B_2^j, B_3^j$.) In what follows  over uniformly over $j\in S$.
First we show that under Condition IV we have with probability $1-o(1)$ for $q\in \{1,2\}$
$$ \| \hat\beta - \beta_0 \|_q \leq  Cs^{1/q}\sqrt{\frac{\log(pKn)}{n}} \ \ \mbox{and} \ \ \max_{j\in S}\|\hat\mu^j-\mu^j_0\|_q \leq Cs^{1/q}\sqrt{\frac{\log(pKn)}{n}} $$
since $|S|\leq p$. The first rate follows from Theorem \ref{thm:linmodel:iv} provided  (i) $\kappa_q^{ZX}(s,3)\geq s^{-1/q}c$ for $q\in\{1,2\}$ and (ii) $\max_{k\in [K]} \En[\xi^2z_{k}^2] \leq C$ with probability $1-o(1)$. Indeed, (i) holds by Condition IV.
Relation (ii) holds by Markov's inequality and Lemma \ref{lem:m2bound}(3) using the relations $\Ep[\max_{i\leq n} \|\xi_iz_{i}\|_\infty^2]\leq K_n^{4/q}$, $\max_{k\in [K]}
\Ep[\frac{1}{n}\sum_{i=1}^n\xi_i^2z_{ik}^2] \leq C$ and $K_n^{4/q}\log K=o(n)$.

The second rate follows from Theorem \ref{thm:nuisancej} provided that we have the following bounds (i) $\kappa_q^{\widetilde Z\widetilde Z}(s,3)\geq s^{-1/q}c$ for $q\in\{1,2\}$, (ii) $\max_{k\leq p} \En[x_{k}^2(z'\mu_0^j)^2] \leq C$ and  (iii) $\max_{k\leq K} \En[v_j^2z_{k}^2] \leq C$
with probability $1-o(1)$. Again, (i) holds by Condition IV.
 Relations (ii) and (iii) hold by the following moment conditions $\max_{\ell \in [p], j\in S} \Ep[\frac{1}{n}\sum_{i=1}^n x_{i\ell}^2(z_i'\mu_0^j)^2] \leq C$,  $\max_{k\in [K], j\in S} \Ep[\frac{1}{n}\sum_{i=1}^n v_{ij}^2z_{ik}^2] \leq C$, $\Ep[\max_{i\leq n, j\in S} \|(v_{ij},z_i',x_i',z_i'\mu_0^j)\|_\infty^4]\leq K_n^{4/q}$ and $K_n^{4/q}\log (pKn)=o(n)$ combined with Markov's inequality and Lemma \ref{lem:m2bound}(3).

Next note that uniformly over $j\in S$
\begin{equation}\label{boundOmegaj}\begin{array}{rl}
|\hat \Omega_j-\Omega_j|& =|\En[x_{j}z'\hat\mu^j]-\Ep[\frac{1}{n}\sum_{i=1}^nx_{ij}z_i'\mu^j_0]|\\
& \leq |\En[x_{j}z'\mu^j_0]+\En[x_{j}z'(\hat\mu^j-\mu^j_0)]| + |\En[x_{j}z'\mu^j_0-\Ep[x_{j}z'\mu^j_0]]| \\
& \leq \|\En[x_{j}z]\|_\infty \|\hat\mu^j-\mu^j_0\|_1  + |\En[x_{j}z'\mu^j_0-\Ep[x_{j}z'\mu^j_0]]|\\
& \leq O(\delta_n \log^{-1/2}(Kpn))\end{array}
\end{equation}
with probability $1-o(1)$ under Condition IV as $\max_{j\in S}\|\En[x_{j}z]\|_\infty \leq C$ and $\max_{j\in S} |\En[x_{j}z'\mu^j_0-\Ep[x_{j}z'\mu^j_0]]| \leq C \{K^{2/q}\log(pKn)/n\}^{1/2}$ with probability $1-o(1)$. Moreover, since $\min{j\in S}|\Omega_j|>c$ by Condition IV, we have that $\min_{j\in S}|\hat\Omega_j|$ is bounded away from zero with probability $1-o(1)$.

Next we proceed to bound each term $B^j_k, k=1,2,3, \ j \in S$. It follows that
$$\begin{array}{rl}
\max_{j\in S}|B_1^j| & = \max_{j\in S}|\hat \Omega_j^{-1}\En[\xi z'](\hat\mu^j-\mu^j_0)|\\
&  \leq \max_{j\in S}|\hat \Omega_j^{-1}|\|\En[\xi z']\|_\infty \|\hat\mu^j-\mu^j_0\|_1\\
& \leq Cs \log(|S|Kpn)/n
\end{array}
$$ since $\|\En[\xi z]\|_\infty \leq \tau \max_{k\leq K} \En[\xi^2z_{k}^2]^{1/2}\leq C\tau$ with probability $1-o(1)$ by (\ref{linmodel:eqsn}) and the argument above, and $\max_{j\in S}\|\hat\mu^j-\mu^j_0\|_1\leq Cs\sqrt{\log(pKn)/n}$ with probability $1-o(1)$.

Moreover we have
$$\begin{array}{rl}
\max_{j\in S}|B_2^j| & = \max_{j\in S}|\hat \Omega_j^{-1}(\hat\beta_{-j}-\beta_{0,-j})'\En[x_{-j}z'\mu^j_0]|\\
&  \leq \max_{j\in S}|\hat \Omega_j^{-1}|\|\En[x_{-j}z'\mu^j_0]\|_\infty \|\hat\beta_{-j}-\beta_{0,-j}\|_1\\
& \leq Cs \log(|S|Kpn)/n
\end{array}
$$
since $\max_{j\in S}\|\En[x_{-j}z'\mu^j_0]\|_\infty \leq \tau \max_{\ell \in [p], j\in S} \En[x_{\ell}^2(z'\mu^j_0)^2]^{1/2}\leq C\tau$ with probability $1-o(1)$ by (\ref{linmodel:eqsn2}) and the argument above, and $\|\hat\beta-\beta_{0}\|_1\leq Cs\sqrt{\log(pKn)/n}$ with probability $1-o(1)$.

Finally, we have
$$\begin{array}{rl}
\max_{j\in S}|B_3^j| & = \max_{j\in S}|\hat \Omega_j^{-1}(\hat\beta_{-j}-\beta_{0,-j})'\En[x_{-j}z'](\hat\mu^j-\mu^j_0)
|\\
&  \leq \max_{j\in S}|\hat \Omega_j^{-1}|\|\hat\beta_{-j}-\beta_{0,-j}\|_1 \|\En[x_{-j}z'](\hat\mu^j-\mu^j_0)\|_\infty\\
& \leq Cs \log(|S|Kpn)/n
\end{array}
$$ since $\|\En[x_{i,-j}z_i'](\hat\mu^j-\mu^j_0)\|_\infty \leq \|\En[x_{-j}z']\hat\mu^j\|_\infty + \|\En[x_{-j}z']\mu^j_0\|_\infty \leq \tau\|\hat t\|_\infty + \tau \|t(\mu^0_j,\vartheta^j_0)\|_\infty \leq C \sqrt{\log(pKn)/n}$ with probability $1-o(1)$ and $\|\hat\beta_{-j}-\beta_{0,-j}\|_1\leq C s\sqrt{\log(pKn)/n}$  with probability $1-o(1)$.
\end{proof}

\begin{proof}[Proof of Corollary \ref{cor:PointwiseInference}]
The result follows directly from the linear representation of Theorem \ref{thm:main} combined with an application of the Lyapunov central limit theorem since $\bar \psi_j(y_i,x_i,z_i)=\Omega^{-1}_j\sigma_j^{-1}(y_i-x_i'\beta_0)z_i'\mu^j_0$, $i=1,\ldots,n$, are independent random variables with  $\Ep[\bar\psi_j(y_i,x_i,z_i)]=0$, $\Ep[\frac{1}{n}\sum_{i=1}^n\bar\psi_j^2(y_i,x_i,z_i)]=1$, and $\Ep[\frac{1}{n}\sum_{i=1}^n|\bar\psi_j(y_i,x_i,z_i)|^3]\leq C$.
\end{proof}

Next we show that the distribution of the (maximum) estimation error is close to the distribution of the maximum of a (sequence of) Gaussian vectors. Let $(\mathcal{G}_j)_{j\in S}$ denote a tight zero-mean Gaussian vector whose dimension $|S|$ could grow with $n$. The covariance matrix associated with $\mathcal{G}$ is given by $V_{jk} = (\Ep[\frac{1}{n}\sum_{i=1}^n\bar \psi_j(y_i,x_i,z_i)\bar\psi_k(y_i,x_i,z_i)])_{j,k}$, $j,k\in S$ where $\bar \psi_j(y_i,x_i,z_i)=\Omega^{-1}_j\sigma_j^{-1}(y_i-x_i'\beta_0)z_i'\mu^j_0$. We have the following lemma.
\begin{lemma}\label{cor:gaussapprox}
In addition to Condition IV, suppose that $A_n^2\log^7(|S|)=o(n)$ and $B_{\bar q,n}^{1/\bar q}\log^{2}(|S|)=o(n^{1/2-1/\bar q})$ where $A_n\geq \max_{j\in S}\Ep[\frac{n^{-1}}{\sigma_j^{3}\Omega_j^{3}}\sum_{i=1}^n|\xi_iz_i'\mu^j_0|^3]$ and $B_{\bar q,n} \geq \frac{1}{n}\sum_{i=1}^n\Ep[ \max_{j\in S}|\xi_iz_i'\mu^j_0|^{\bar q}]$ for some $\bar q>3$. Then we have that
$$ \sup_{t \in \mathbb{R}} \left|P\left(\max_{j\in S}|\sqrt{n}\sigma_j^{-1}(\check\beta_{0j}-\beta_{0j})|\leq t\right)- P\left(\max_{j\in S}|\mathcal{G}_j|\leq t\right) \right|=o(1)$$
\end{lemma}
\begin{proof}[Proof of Lemma \ref{cor:gaussapprox}]

It is convenient to define
 $$\begin{array}{rl}
 Z &= \max_{j\in S} |\sqrt{n}\sigma_j^{-1}(\check\beta_j - \beta_{0j})|,\\
 \bar Z &= \max_{j\in S} \left|\frac{1}{\sqrt{n}}\sum_{i=1}^n\bar\psi_j(y_i,x_i,z_i)\right|, \ \ \mbox{and}\\
 \widetilde Z &= \max_{j\in S} |\mathcal{G}_j|\end{array}$$ where $(\mathcal{G}_j)_{j\in S}$ is a Gaussian vector with zero mean and covariance matrix given by $V_{jk}=\Ep[n^{-1}\sum_{i=1}^n\bar\psi_j(y_i,x_i,z_i)\bar\psi_k(y_i,x_i,z_i)]$, $j,k\in S$.

By the triangle inequality we have
$$\begin{array}{rl}
|Z - \widetilde Z| & \leq |Z-\bar Z| + |\bar Z - \widetilde Z| \\
& \leq \max_{j \in S} \left| \sqrt{n}\sigma_j^{-1}(\check\beta_j - \beta_{0j}) -\frac{1}{\sqrt{n}}\sum_{i=1}^n\bar\psi_j(y_i,x_i,z_i) \right| + |\bar Z - \widetilde Z| \\
& \leq O_\P(\delta_n)  + |\bar Z - \widetilde Z|
\end{array}
$$ where the last inequality follows from the linear representation established in Theorem \ref{thm:main} under Condition IV.

To bound $|\bar Z - \widetilde Z|$ we will apply Theorem 3.1 in \cite{chernozhukov2015noncenteredprocesses}. Let $X_{ij}=\sigma_j^{-1}\Omega_j^{-1}(y_i-x_i'\beta_0)z_i'\mu_0^j$, $\Ep[X_{ij}]=0$, $Y_i\sim N(0,\Ep[X_iX_i'])$, $L_n=\max_{j\in S}\Ep[\frac{1}{n}\sum_{i=1}^n|X_{ij}|^3]$,
$M_{n,X}(\delta) = \frac{1}{n}\sum_{i=1}^n\Ep[ \max_{j\in S}|X_{ij}|^31\{ \max_{j\in S}|X_{ij}|> \delta \sqrt{n}/\log |S| \} ]$, and $M_{n,Y}(\delta) = \frac{1}{n}\sum_{i=1}^n\Ep[ \max_{j\in S}|Y_{ij}|^31\{ \max_{j\in S}|Y_{ij}|> \delta \sqrt{n}/\log |S| \} ]$. Thus $\bar Z = \max_{j\in S} |n^{-1/2}\sum_{i=1}^nX_{ij}|$ and $\widetilde Z =  \max_{j\in S} |n^{-1/2}\sum_{i=1}^nY_{ij}|$. Theorem 3.1 in \cite{chernozhukov2015noncenteredprocesses} shows that for every Borel set $A\subseteq \mathbb{R}$ such that
$$\begin{array}{rl} P ( \max_{j\in S} n^{-1/2}\sum_{i=1}^nX_{ij} \in A ) & \leq P(\max_{j\in S} n^{-1/2}\sum_{i=1}^nY_{ij} \in A^{C\delta} )\\
 & +  C' \frac{\log^2|S|}{\delta^3\sqrt{n}}\{L_n+M_{n,X}(\delta)+M_{n,Y}(\delta)\}\end{array}$$
where $A^{C\delta}=\{ t \in \mathbb{R} : {\rm dist}(t,A)\leq C\delta\}$ is an $C\delta$-enlargement of $A$, and  $|S|>2$. By definition of $A_n$ we have $$L_n=\max_{j\in S}\Ep[\mbox{$\frac{1}{n}\sum_{i=1}^n$}|X_{ij}|^3] = \max_{j\in S}\Ep[\mbox{$\sigma_j^{-3}\Omega_j^{-3}\frac{1}{n}\sum_{i=1}^n$} |\xi_iz_i'\mu_0^j|^3] \leq A_n.$$

Since $\frac{1}{n}\sum_{i=1}^n\Ep[ \max_{j\in S} |\xi_iz_i'\mu^j_0|^{\bar q} ] \leq B_{\bar q,n}$, we have that
$$\begin{array}{rl}
M_{n,X}(\delta) & = \frac{1}{n}\sum_{i=1}^n\Ep[ \max_{j\in S}|X_{ij}|^3 1\{ \max_{j\in S}|X_{ij}|> \delta \sqrt{n}/\log |S| \} ]\\
 & \leq \frac{1}{n}\sum_{i=1}^n\Ep[ \max_{j\in S}|X_{ij}|^3 \max_{j\in S}|X_{ij}|^{\bar q-3}/ \{ \delta \sqrt{n}/\log |S| \}^{\bar q-3} ]\\
 & \leq \frac{1}{n}\sum_{i=1}^n\Ep[ \max_{j\in S}|X_{ij}|^{\bar q} ]/ \{ \delta \sqrt{n}/\log |S| \}^{\bar q-3}\\
 & \leq B_{\bar q,n} \{\log(|S|) / (\delta \sqrt{n})\}^{\bar q-3}\\
 \end{array}$$

Next note that
$$\begin{array}{rl}
M_{n,Y}(\delta) & = \frac{1}{n}\sum_{i=1}^n\Ep[ \max_{j\in S}|Y_{ij}|^3 1\{ \max_{j\in S}|Y_{ij}|> \delta \sqrt{n}/\log |S| \} ]\\
 & \leq \frac{1}{n}\sum_{i=1}^n\Ep[ \max_{j\in S}|Y_{ij}|^{\bar q}] / \{ \delta \sqrt{n}/\log |S| \}^{\bar q-3} \\
 & \leq \frac{C_{\bar q}\log^{\bar q/2}(|S|)}{n}\sum_{i=1}^n \max_{j\in S} \Ep[ |Y_{ij}|^2]^{\bar q/2} / \{ \delta \sqrt{n}/\log |S| \}^{\bar q-3} \\
 & = \frac{C_{\bar q}\log^{\bar q/2}(|S|)}{n}\sum_{i=1}^n \max_{j\in S} \Ep[ |X_{ij}|^2]^{\bar q/2} / \{ \delta \sqrt{n}/\log |S| \}^{\bar q-3} \\
 & \leq C_{\bar q}\log^{\bar q/2}(|S|) B_{\bar q,n} \{\log(|S|) / (\delta \sqrt{n})\}^{\bar q-3}\\
 \end{array}$$


Setting $\delta:=  \frac{A_n^{1/3}\log^{2/3}(|S|)}{\gamma^{1/3}n^{1/6}} + \frac{B_{\bar q,n}^{1/\bar q} \log^{3/2-1/\bar q}(|S|)}{\gamma^{1/\bar q}n^{1/2-1/\bar q}}$, we have
$$ P ( |\bar Z - \widetilde Z| > C\delta ) \leq  C' \gamma + C'/n.$$

By Strassen's Theorem, there exists a version of $\widetilde Z$ such that
$$
\left| \bar Z - \widetilde Z \right| = O_P\left( \frac{A_n^{1/3}\log^{2/3}(|S|)}{n^{1/6}} + \frac{B_{\bar q,n}^{1/\bar q} \log^{3/2-1/\bar q}(|S|)}{n^{1/2-1/\bar q}} \right)
$$
The result then follows from Lemma 2.4 in \cite{chernozhukov2012gaussian} so that $\sup_{t\in \mathbb{R}}|P( Z \leq t ) -   P(\widetilde Z \leq t )|  = o(1)$ by noting that $\Ep[\widetilde Z] \leq C\sqrt{\log |S|}$ since $\Ep[\mathcal{G}_j^2]=1$, provided that
$$\sqrt{\log |S|} \left(  \delta_n +  \frac{A_n^{1/3}\log^{2/3}(|S|)}{n^{1/6}} + \frac{B_{\bar q,n}^{1/\bar q} \log^{3/2-1/\bar q}(|S|)}{n^{1/2-1/\bar q}}\right)=o(1),$$
which holds under $A_n^2\log^7(|S|)=o(n)$ and $B_{\bar q,n}^{1/\bar q}\log^{2}(|S|)=o(n^{1/2-1/\bar q})$.
\end{proof}

\begin{proof}[Proof of Theorem \ref{thm:inference}]
Let $c_\alpha^*$ denote the $(1-\alpha)$-conditional quantile of $\widetilde Z^* = \max_{j\in S} |\widehat{\mathcal{G}}_j|$ and
 $c_\alpha^0$ denote the $(1-\alpha)$-conditional quantile of $\widetilde Z= \max_{j\in S} |{\mathcal{G}}_j|$ where $\mathcal{G}$ is a Gaussian random vector defined in Lemma \ref{cor:gaussapprox}. For some sequence $\vartheta_n \to 0$ defined in Step 3 below we have that
$$ \begin{array}{rl}
&\displaystyle \Pr{\max_{j\leq S}|\sqrt{n}\hat\sigma_j^{-1}(\check\beta_{0j}-\beta_{0j})|\leq c_\alpha^*} \\
& = \Pr{ |\sqrt{n}\sigma_j^{-1}(\check\beta_{0j}-\beta_{0j})|\leq c_\alpha^* \sigma_j/\hat\sigma_j, \forall j\in S} \\
& \leq_{(1)}  \Pr{ \max_{j\in S}|\sqrt{n}\sigma_j^{-1}(\check\beta_{0j}-\beta_{0j})|\leq c_\alpha^* (1+\varepsilon_n)}+o(1)\\
& \leq_{(2)}  \Pr{ \max_{j\in S}|\sqrt{n}\sigma_j^{-1}(\check\beta_{0j}-\beta_{0j})|\leq c^0_{\alpha-\vartheta_n}}+o(1)\\
& \leq_{(3)}  \Pr{ \max_{j\in S}|\mathcal{G}_j|\leq c^0_{\alpha-\vartheta_n}}+o(1)\\
& = 1-\alpha+\vartheta_n+o(1)
\end{array}$$
where inequality (1) holds by $|\sigma_j/\hat\sigma_j|\leq (1+\varepsilon_n)$ with probability $1-o(1)$ established in Step 2, (2) follows from $c_\alpha^* (1+\varepsilon_n) \leq c^0_{\alpha-\vartheta_n}$ with probability $1- o(1)$ by Step 3 below, (3) follows by Lemma \ref{cor:gaussapprox}, and the last step by definition of $c^0_{\alpha-\vartheta_n}$.

The reverse inequality holds by a similar argument.

~\\

Step 2. Claim: for some $\varepsilon_n = O(\delta_n\log^{1/2}(n|S|))$ we have \begin{equation}\label{step2:main}P\left( \max_{j\in S} |\sigma_j/\hat\sigma_j|\vee |\hat \sigma_j/\sigma_j|\leq 1+\varepsilon_n\right) \geq 1-o(1)\end{equation}
Let $\widetilde\psi_j(y,x,z)=\psi_j(y,x,z,\hat\beta,\hat\mu^j)$ and $\psi_j(y,x,z)=\psi_j(y,x,z,\beta_{0},\mu^j_0)$, so that $\sigma_j=\Omega_j^{-1}\{\Ep[\frac{1}{n}\sum_{i=1}^n\widetilde\psi_j^2(y_i,x_i,z_i)]\}^{1/2}$, we have that
$$
\begin{array}{rl}
|\hat\sigma_j - \sigma_j| & \leq \hat\Omega_j^{-1}|  \En[\widetilde\psi^2_j(y,x,z)]^{1/2} - \En[\psi^2_j(y,x,z)]^{1/2}\\
&\ \ \ \ \ \ \ \ +  \En[\psi^2_j(y,x,z)]^{1/2} -  \Ep[\frac{1}{n}\sum_{i=1}^n\psi^2_j(y_i,x_i,z_i)]^{1/2}| \\
& + |\hat\Omega_j^{-1}-\Omega_j^{-1}|\Ep[\frac{1}{n}\sum_{i=1}^n\psi^2_j(y_i,x_i,z_i)]^{1/2}\\
& \leq \hat\Omega_j^{-1}\En[\{\widetilde\psi_j(y,x,z)-\psi_j(y,x,z)\}^2]^{1/2}\\
& + \hat\Omega_j^{-1}\frac{|\En[\psi^2_j(y,x,z)] -  \Ep[\frac{1}{n}\sum_{i=1}^n\psi_j^2(y_i,x_i,z_i)]| }{\En[\psi^2_j(y,x,z)]^{1/2} +  \Ep[\frac{1}{n}\sum_{i=1}^n\psi^2_j(y_i,x_i,z_i)]^{1/2}}\\
& + \frac{|\Omega_j-\hat\Omega_j|}{\hat\Omega_j\Omega_j}\Ep[\frac{1}{n}\sum_{i=1}^n\psi^2_j(y_i,x_i,z_i)]^{1/2}\\
& =\hat\Omega_j^{-1}\{ T_1^j + T_2^j + T_3^j \}\end{array}
$$

By the definitions $K_n=\Ep[\max_{i\leq n} \|W_i\|_\infty^q]^{1/q}$, $\widetilde\psi_j(y,x,z)$ and $\psi_j(y,x,z)$, and the triangle inequality we have uniformly over $j\in S$
\begin{equation}\label{Aux:L2Pn}
 \begin{array}{rl}
|T_1^j|  &= \En[\{\widetilde\psi_j(y,x,z)-\psi_j(y,x,z)\}^2]^{1/2} \\
&  \leq \En[\{(\hat\mu^j-\mu^j_0)'z\}^2\xi^2]^{1/2}  + \En[\{x'(\hat\beta-\beta_0)\}^2(z'\hat\mu^j)^2]^{1/2}\\
& \leq C\sqrt{\frac{s\log(Kpn)}{n}}\max_{i\leq n}|\xi_i|\\
 & + C\sqrt{\frac{s\log(Kpn)}{n}}  {\displaystyle \max_{i\leq n, j\in S}}\{|z_i'\mu^j_0|+|z_i'(\hat\mu^j-\mu^j_0)|\}\\
& = O_\P\left(K_n^{1/q}\sqrt{\frac{s\log(Kpn)}{n}}+\sqrt{\frac{s\log(Kpn)}{n}}K_n^{1/q}s\sqrt{\frac{\log(pKn)}{n}}\right)
\end{array}\end{equation}
where we used that $\max_{i\leq n, j\in S}\|(\xi_i,x_i,z_{i},z_i'\mu^j_0)\|_\infty = O_\P(K_n^{1/q})$, the (prediction) rate of convergence for $\hat\beta$ and $\hat\mu^j, j\in S$, and the condition on $s$, $n$, $p$, $K$, and $\delta_n$ in Condition IV.

Next we bound $\max_{j\in S}|T_2^j|$. Let $X_{ij}:= \psi_j(y_i,x_i,z_i)\sigma_j^{-1}\Omega_j^{-1}$, $j\in S$, $p=|S|\geq 2$, $k=2$, and note that $\Ep[\frac{1}{n}\sum_{i=1}^n\psi_j^2(y_i,x_i,z_i)]=\sigma_j^2\Omega_j^2$. Thus by  Lemma \ref{lem:m2bound}(3)  we have
$$ \begin{array}{rl}
\Ep[\max_{j\in S}|T_2^j|]&\leq \Ep\left[  \max_{j\in S}|\En[\psi^2_j(y,x,z)]-\Ep[\frac{1}{n}\sum_{i=1}^n\psi_j^2(y_i,x_i,z_i)]|/\sigma_j^2\Omega_j^2 \right] \\
  & \lesssim \sqrt{\frac{\log(|S|)}{n}} M_{k}^{1/2}+M_k\frac{\log(|S|)}{n}\\
& \lesssim \frac{\log^{1/2}(|S|) K_n^{2/q}}{\sqrt{n}\min_{j\in S} \sigma_j\Omega_j}\\
\end{array}$$
\noindent where we used that $M_{k} =\Ep[ \max_{i\leq n} \|X_i\|_\infty^2 ] \leq K_n^{4/q}\sigma_j^{-2}\Omega_j^{-2}$ by Condition IV. Thus by Markov's inequality we have
$$ \max_{j\in S}|T_2^j| = O_\P(\frac{\log^{1/2}(|S|) K_n^{2/q}}{\sqrt{n}\min_{j\in S} \sigma_j\Omega_j}) = O_P(\delta_n \log^{1/2}(|S|))$$
under $K_n^{2/q} \leq \delta_n \sqrt{n}\min_{j\in S} \sigma_j\Omega_j$ implied by Condition IV.


For the last term, we have with probability $1-o(1)$
$$\begin{array}{rl}
&\max_{j\in S}|\hat\Omega_j-\Omega_j|\frac{\Ep[\frac{1}{n}\sum_{i=1}^n\psi^2_j(y_i,x_i,z_i)]^{1/2}}{\Omega_j} \\
 & \leq \max_{j\in S}|\En[x_jz'(\hat\mu^j-\mu_0^j)]|\sigma_j +|\En[x_{j}z_{-j}'\mu^j_0]-\frac{1}{n}\sum_{i=1}^n\Ep[x_{ij}z_{i}'\mu^j_0]|\\
& \leq \max_{j\in S}\{\En[x_j^2]\}^{1/2}\{\En[\{z_{-j}'(\hat\mu^j-\mu_0^j)\}^2]\}^{1/2}\sigma_j  \\
& + \max_{j\in S}|\En[x_{j}z_{-j}'\mu^j_0]-\frac{1}{n}\sum_{i=1}^n\Ep[x_{ij}z_{i}'\mu^j_0]|\sigma_j\\
& \lesssim  \sqrt{n^{-1} s\log(pKn)}+\sqrt{n^{-1}\log(|S|)}\bar m_2  +\sqrt{n^{-1}\log(n)} \bar m_2 + n^{-1}K_n^2 \log n  \\
& \lesssim \delta_n (1+\max_{j\in S}\|\mu^j_0\|)\end{array}$$
\noindent where $\bar m_2 := \max_{j\in S}\{\frac{1}{n}\sum_{i=1}^n\Ep[x_{ij}^2(z_{i}'\mu^j_0)^2|\}^{1/2}\leq C\max_{j\in S} \|\mu^j_0\|$, and we used that  $\Ep[\max_{i\leq n, j\in S} |x_{ij}z_i'\mu^j_0|^2]^{1/2} \leq K_n^{2/q}$. Note that the first step follows from $\Ep[\frac{1}{n}\sum_{i=1}^n\psi^2_j(y_i,x_i,z_i)]^{1/2}=\sigma_j\Omega_j\leq C$ and the triangle inequality,  the second step from Cauchy--Schwarz inequality, the third from the (prediction) rate of convergence that holds with probability $1-o(1)$, Lemma \ref{lem:m2bound}(2) and Lemma \ref{lemma:CCK}. The last step follows from the growth conditions in Condition IV.
Since  $\Omega_j$ is bounded away from zero and from above, the last argument also implies that $\hat\Omega_j$ is also bounded from below uniformly in $j\in S$ and $n$ with probability $1-o(1)$.

Combining these relations we have $$\varepsilon_n \leq C\delta_n\log^{1/2}(n|S|)(1+\max_{j\in S}\|\mu^j_0\|).$$

Step 3. Claim: there is a sequence $\vartheta_n \to 0$ such that
$$ P( c_\alpha^*(1+\varepsilon_n) > c^0_{\alpha-\vartheta_n}) = o(1)$$
where $\varepsilon_n$ is defined in Step 2.

Recall that $\widehat{\mathcal{G}}_j=\frac{1}{\sqrt{n}} \sum_{i=1}^ng_i\hat\psi_{ij}$ where $\hat\psi_{ij} :=\hat\sigma_j^{-1}\hat\Omega_j^{-1}(y_i-x_i'\hat\beta)z_i'\hat\mu^j=\hat\sigma_j^{-1}\hat\Omega_j^{-1}\widetilde \psi_j(y_i,x_i,z_i,\check\beta_j,\hat\beta_{-j},\hat\mu^j)$ and define
 $$
\bar Z^* = \max_{j\in S} \left|\frac{1}{\sqrt{n}} \sum_{i=1}^ng_i\bar\psi_j(y_i,x_i,z_i)\right|, \ \  \widetilde Z^*= \max_{j\in S} |\widehat{\mathcal{G}}_j|, \ \ \mbox{and} \ \
 \widetilde Z  = \max_{j\in S} |\mathcal{G}_j|$$ where $(\mathcal{G}_j)_{j\in S}$ is a zero mean Gaussian vector with covariance matrix given by $C_{jk}=\Ep[\frac{1}{n}\sum_{i=1}^n\bar\psi_j(y_i,x_i,z_i)\bar\psi_k(y_i,x_i,z_i)]$, and $\widetilde Z^*$ is associated with the multiplier bootstrap as defined in (\ref{def:mb}).

 We have that
\begin{equation}\label{aux:step3}\begin{array}{rl}
 |\widetilde Z^* - \widetilde Z | & \leq \left|\widetilde Z^* -  \bar Z^* \right| + \left| \bar Z^* - \widetilde Z \right|  \\
& \leq {\displaystyle \max_{j\in S}}  \left| \frac{1}{\sqrt{n}} \sum_{i=1}^ng_i\{\hat\psi_{ij} - \bar\psi_j(y_i,x_i,z_i)\}\right| + \left|\bar Z^* - \widetilde Z\right| \\ \end{array}
 \end{equation}

To control the first term of the RHS in (\ref{aux:step3}) note that conditional on $(y_i,x_i,z_i)_{i=1}^n$, $\frac{1}{\sqrt{n}} \sum_{i=1}^ng_i\{\hat\psi_{ij} - \bar\psi_j(y_i,x_i,z_i)\}$ is a zero-mean Gaussian random variable with variance $\En[\{\hat\psi_{ij}-\bar\psi_j(y,x,z)\}^2]$. By Step 2's claim (\ref{step2:main}) and (\ref{Aux:L2Pn}), with probability $1-o(1)$ we have uniformly over $j\in S$
$$\begin{array}{rl}\En[\{\hat\psi_{j}-\bar\psi_j(y,x,z)\}^2]^{1/2} & \leq |\hat\sigma_j^{-1}\hat\Omega_j^{-1}-\sigma_j^{-1}\Omega_j^{-1}|\En[\psi_j^2(y,x,z,\check\beta_j,\hat\eta^j)]^{1/2} \\
& + \sigma_j^{-1}\Omega_j^{-1}\En[\{\widetilde \psi_j(y,x,z)-\psi_j(y,x,z)\}^2]^{1/2}\\
& \lesssim \varepsilon_n \\
\end{array}$$
from the calculations in Step 2, since $\hat\sigma_j$, $\sigma_j$, $\Omega_j$, $\hat\Omega_j$ and $\En[\hat\psi_j^2(y,z)]^{1/2}$ are bounded away from zero and from above with probability $1-o(1)$. Therefore we have with probability $1-o(1)$ that
$$ \begin{array}{rl}\Ep\left[ \max_{j\in S} \left| \frac{1}{\sqrt{n}} \sum_{i=1}^ng_i\{\hat\psi_j(y_i,x_i,z_i) - \bar\psi_j(y_i,x_i,z_i)\}\right| \ \mid (y_i,x_i,z_i)_{i=1}^n \right] \\
\lesssim \varepsilon_n \sqrt{\log |S| } =: I_n
\end{array} $$
by Corollary 2.2.8 in \cite{vdV-W}.

To control the second term of the RHS in (\ref{aux:step3}) we apply Theorem 3.2 in \cite{chernozhukov2015noncenteredprocesses} and a conditional version of Strassen's theorem due to \cite{monrad1991nearby}. We have $\bar Z^*=\max_{j\in S} |X_j|$ and $\widetilde Z=\max_{j\in S}|Y_j|$ where $X_{j}=\frac{1}{\sqrt{n}}\sum_{i=1}^ng_i\bar \psi_j(y_i,x_i,z_i)$, $j\in S$, and $Y \sim N(0,\Ep[\frac{1}{n}\sum_{i=1}^n\bar\psi(y_i,x_i,z_i)\bar\psi'(y_i,x_i,z_i)])\in \mathbb{R}^{|S|}$. Note that $\Ep[X_j]=\Ep[Y_j]=0$ and $\Ep[X_j^2]=\Ep[Y_j^2]=1$ for all $j\in S$. Let $$\Delta=\max_{j,k\in S}| \En[\bar\psi_j(y,x,z)\bar\psi_k(y,x,z)] - \Ep[\mbox{$\frac{1}{n}\sum_{i=1}^n$}\bar\psi_j(y_i,x_i,z_i)\bar\psi_k(y_i,x_i,z_i)]|.$$
 We will consider the event $E_n=\{ \Delta \leq \bar \Delta \}$.  Conditionally on $E_n$, by  Theorem 3.2 in \cite{chernozhukov2015noncenteredprocesses} we have that for every $\delta>0$ and every Borel subset $A\subset \mathbb{R}$
 $$ P( \max_{j\in S} X_j \in A \mid (y_i,x_i,z_i)_{i=1}^n) \leq P(\max_{j\in S} Y_j \in A^\delta ) + C\delta^{-1}\sqrt{\bar \Delta \log(|S|)}$$
for some universal constant $C>0$ where $A^{\delta}=\{ t \in \mathbb{R} : {\rm dist}(t,A)\leq \delta\}$. By a conditional version of Strassen's theorem, we have there is a version of $\widetilde Z$  such that $$P( | \bar Z^* - \widetilde Z | > \delta ) \leq P(E_n^c) + C\delta^{-1}\sqrt{\bar \Delta \log(|S|)}.$$
The result will follow from a suitable choice of $\bar \Delta \to 0$ and $\delta \to 0$. We have that
$$
\begin{array}{rl}
\Ep[\Delta] & \leq C \Ep[  \Ep_{r} [ \En[r_i\bar \psi_{ij}\bar\psi_{ik}] ] \leq C\sqrt{n^{-1}\log(|S|^2)} \Ep[ \max_{j,k\in S} \En[  \bar \psi_{ij}^2\bar\psi_{ik}^2]^{1/2}]\\
&  = C\sqrt{n^{-1}\log(|S|^2)}  \Ep[ \max_{j,k\in S} \En[  \xi_i^4 |z_i'\mu^j_0z_i'\mu^k_0|^2]^{1/2}]\\
&  \leq C\sqrt{n^{-1}\log(|S|^2)}  \Ep[ \max_{j,k\in S, i\leq n} |z_i'\mu^j_0z_i'\mu^k_0| \En[  \xi_i^4 ]^{1/2}]\\
&  \leq C\sqrt{n^{-1}\log(|S|^2)}  \Ep[ \max_{j,k\in S, i\leq n} |z_i'\mu^j_0z_i'\mu^k_0|^2]^{1/2} \{\Ep[\frac{1}{n}\sum_{i=1}^n \xi_i^4 ]\}^{1/2}\\
& \leq C\sqrt{n^{-1}\log(|S|^2)} K_n^{2/q}
\end{array}
$$
For $\gamma \in (0,1)$, we can set $\bar \Delta = \gamma^{-1}\sqrt{n^{-1}\log(|S|^2)} K_n^{2/q}$, we have $P(E_n^c)=O(\gamma)$ and by setting $\delta = \gamma^{-1}\{\bar\Delta \log |S|\}^{1/2}$ we have
{\small $$ \begin{array}{rl}
\displaystyle |\bar Z^* - \widetilde Z | &\displaystyle  = O_P\left( n^{-1/4}\log^{3/4}(|S|) K_n^{1/q} \right) =: II_n
\end{array} $$}
where $II_n\to 0$ under $n^{-1}K_n^{4/q}\log^5(|S|)=o(1)$ under Condition IV since $q\geq 4$.

For notational convenience let $r_n:= \ell_n(I_n + II_n)\to 0$ for some $\ell_n\to \infty$, and $\wp_n = P( |\widetilde Z^* - \widetilde Z| > r_n )^{1/2} =o(1)$.
Letting $U_n := P( |\widetilde Z^* - \widetilde Z| > r_n \mid (y_i,x_i,z_i)_{i=1}^n)$ note that $$\begin{array}{rl}
P(  U_n > \wp_n ) & = E [ 1\{ U_n > \wp_n \} ]\\
& \leq  E[ U_n ] / \wp_n \\
& = \wp_n\\
\end{array}
$$ so that $P( |\widetilde Z^* - \widetilde Z| > r_n \mid (y_i,x_i,z_i)_{i=1}^n) \leq \wp_n$ with probability at least $1-\wp_n$. Then, by definition of the quantile function we have that
$$\begin{array}{rl}
(1+\varepsilon_n)c_\alpha^* & \leq (1+\varepsilon_n)(c^0_{\alpha - \wp_n}+r_n) \\
& = c^0_{\alpha-\vartheta_n}- \{ c^0_{\alpha-\vartheta_n} - c^0_{\alpha - \wp_n}\} +c^0_{\alpha - \wp_n} \varepsilon_n + (1+\varepsilon_n)r_n\\
& \leq c^0_{\alpha-\vartheta_n} - \frac{c(\vartheta_n- \wp_n)}{\Ep[\widetilde Z]} +  c^0_{\alpha - \wp_n}\varepsilon_n + (1+\varepsilon_n)r_n\\
& \leq c^0_{\alpha-\vartheta_n}
\end{array}
$$ where the third step we used Corollary 2.1 in \cite{chernozhukov2014honestbands}, and we set $\vartheta_n\to 0$ so that $\Ep[\widetilde Z] \{  c^0_{\alpha - \wp_n}\varepsilon_n + r_n(1+\varepsilon_n) \} = o(\vartheta_n- \wp_n)$. Note that because  $\Ep[\widetilde Z] \leq C\sqrt{1 + \log |S|}$ and $ c^0_{\alpha - \wp_n} \leq C\sqrt{1 + \log (|S|/\{\alpha - \wp_n\})}$, we can find such sequence $\vartheta_n\to 0$ provided that
we have $$\varepsilon_n (1 + \log |S|) = O(\delta_n \log^{1/2}(n|S|)(1+\max_{j\in S}\|\mu^j_0\|)\log |S| = o(1),$$  $r_n = o(1)$  and also $$\begin{array}{rl}
r_n \sqrt{\log |S|} & = \ell_n O( \varepsilon_n\sqrt{\log |S|} + n^{-1/4}\log^{3/4}(|S|) K_n^{1/q}) \sqrt{\log |S|} \\
& = \ell_n O( \delta_n \log^{3/2}(n|S|) ) + \ell_n O( n^{-1/4}K^{1/q}_n\log^{5/4}(|S|))\\
& = o(1)
\end{array}$$
implied by Condition IV where $\delta_n \log^{3/2}(n|S|)=o(1)$, $\max_{j\in S}\|\mu^j_0\|\leq C$ and $ n^{-1}K_n^{4/q} \log^5(|S|) = o(1)$ are assumed.

\end{proof}

\section{Auxiliary Lemmas}

\begin{lemma}\label{lem:m2bound}
Let $X_i, i=1,\ldots,n,$ be independent random vectors in $\mathbb{R}^p$, $p\geq 3$. Define $\bar m_k := \max_{j\leq p}\frac{1}{n}\sum_{i=1}^n\mathbb{E}[|X_{ij}|^k]$ and $M_{k} \geq \mathbb{E}[ {\displaystyle \max_{i\leq n}}\|X_i\|_\infty^k]$. Then we have the following bound
$$\begin{array}{ll}
(1) \mathbb{E}\left[\max_{j\leq p}\frac{1}{n}\sum_{i=1}^n|X_{ij}|\right] \leq CM_{1} n^{-1}\log p+ C\bar m_1 \\
(2)\mathbb{E}\left[\max_{j\leq p}\frac{1}{n}\left|\sum_{i=1}^n X_{ij}-\mathbb{E}[X_{ij}]\right|\right] \leq C  \sqrt{\bar m_2 n^{-1}\log p} + M_2^{1/2} n^{-1}\log p\\
(3)\mathbb{E}\left[\max_{j\leq p}\frac{1}{n}\left|\sum_{i=1}^n|X_{ij}|^k-\mathbb{E}[|X_{ij}|^k]\right|\right] \leq \frac{C M_{k}\log p}{n}+C\sqrt{\frac{M_{k}\bar m_k \log p}{n}}\\
\end{array}
$$
for some universal constant $C$.
\end{lemma}

The following lemma is a version of a result in \cite{RudelsonVershynin2008}.

\begin{lemma}\label{thm:RV34}
Let $(X_{i})$, $i=1,\ldots, n$, be independent (across i) random vectors such that $X_{i} \in \mathbb{R}^p$ with $p\geq 3$ and $(\Ep[ \max_{1\leq i\leq n}\|X_{i}\|_\infty^2])^{1/2} \leq K$. Furthermore, for $k\geq 1$, define
$$
\delta_n:= \frac{K \sqrt{k}}{\sqrt n}\left(\log^{1/2} p + (\log k) (\log^{1/2}p) (\log^{1/2} n) \right),
$$
Then, there is a universal constant $C$ such that
{\small $$
\begin{array}{c}
\Ep\left[ \sup_{\|\theta\|_0\leq k, \|\theta\| =1} \left| \frac{1}{n}\sum_{i=1}^n\{ (\theta'X_i)^2 - \Ep[(\theta'X_i)^2]\} \right|\right]\\ \leq C\delta_n^2 + C\delta_n \sup_{\|\theta\|_0\leq k, \|\theta\| =1} \sqrt{\frac{1}{n}\sum_{i=1}^n\Ep[(\theta'X_i)^2]}
\end{array}
$$}
\end{lemma}

Let $ ( W_i)_{i=1}^n$ be a sequence of independent copies of a random element $ W$  taking values in a measurable space $({\mathcal{W}}, \mathcal{A}_{{\mathcal{W}}})$ according to a probability law $P$. Let $\mathcal{F}$ be a set  of suitably measurable functions $f\colon {\mathcal{W}} \to \mathbb{R}$, equipped with a measurable envelope $F\colon \mathcal{W} \to \mathbb{R}$. Let $\Gn(f)=n^{-1/2}\sum_{i=1}^n f(W_i)-\Ep[f(W_i)]$.

  \begin{lemma}[Maximal Inequality adapted from \cite{chernozhukov2012gaussian}]
\label{lemma:CCK}  Suppose that $F\geq \sup_{f \in \mathcal{F}}|f|$ is a measurable envelope for the finite class $\mathcal{F}$
with $\| F\|_{P,\bar q} < \infty$ for some $\bar q \geq 2$.  Let $M = \max_{i\leq n} F(W_i)$ and $\sigma^{2} > 0$ be any positive constant such that $\sup_{f \in |\mathcal{F}|}  \| f \|_{P,2}^{2} \leq \sigma^{2} \leq \| F \|_{P,2}^{2}$. Then
\begin{equation*}
\Ep_P [ \max_{f \in \mathcal{F}} |\Gn(f)| ] \leq K  \left( \sqrt{\sigma^{2} \log \left ( \frac{|\mathcal{F}|\| F \|_{P,2}}{\sigma} \right ) } + \frac{\| M \|_{P, 2}}{\sqrt{n}} \log \left ( |\mathcal{F}|  \frac{\| F \|_{P,2}}{\sigma} \right ) \right),
\end{equation*}
where $K$ is an absolute constant.  Moreover, for every $t \geq 1$, with probability $> 1-t^{-q/2}$,
\begin{multline*}
\max_{f \in \mathcal{F}} |\Gn(f)|  \leq (1+\alpha) \Ep_P [ \max_{f \in \mathcal{F}} |\Gn(f)|  ] \\
+ K(\bar q) \Big [ (\sigma + n^{-1/2} \| M \|_{P,\bar q}) \sqrt{t}
+  \alpha^{-1}  n^{-1/2} \| M \|_{P,2}t \Big ], \ \forall \alpha > 0,
\end{multline*}
where $K(\bar q) > 0$ is a constant depending only on $\bar q$.
\end{lemma}

\end{appendix}

\bibliography{IVbib}
\bibliographystyle{plain}

\end{document}